\documentclass[11pt]{article}
\usepackage{xcolor}
\usepackage[numbers,sort&compress]{natbib}
\usepackage{graphicx}
\usepackage{subfig}

\usepackage{tikz}
\usetikzlibrary{arrows,positioning} 
\usepackage{pgfplots}
\usepackage{pgfplotstable}
\pgfplotsset{compat=newest}
\usepackage{amsmath}
\usepackage{amssymb}
\usepackage{amsthm}
\usepackage{euscript}
\usepackage{epic,eepic}
\usepackage{epsf}
\usepackage{dsfont}
\definecolor{labelkey}{rgb}{0,0.08,0.45}
\definecolor{refkey}{rgb}{0,0.6,0.0}
\definecolor{Brown}{rgb}{0.45,0.0,0.05}
\definecolor{dgreen}{rgb}{0.00,0.49,0.00}
\definecolor{dblue}{rgb}{0,0.08,0.75}
\title{
Reinforcement learning with restrictions on the action set
}
\author{ Mario Bravo$^1$ and Mathieu Faure$^2$
\\[5mm]
\small $\!^1$ Instituto de Sistemas Complejos de Ingenier\'ia (ISCI),\\
\small Universidad de Chile \\
 \small Rep\'ublica 701, Santiago, Chile.\\
 {\ttfamily mbravo@dii.uchile.cl}\\
[4mm] \small $\!^2$Aix-Marseille University (Aix-Marseille School of Economics)\\
\small CNRS \& EHESS \\
\small  2 rue de la vieille charit\'e 13236 Marseille Cedex 02\\
{\ttfamily mathieu.faure@univ-amu.fr}
}

\newcommand{\escon}[2]{\mathbb E (#1 \, \,  \vert \, \,  #2)}
\newcommand{\pcon}[2]{\mathbb P (#1 \, \,  \vert \, \,  #2)}

\newcommand{\RR}{\ensuremath{\mathbb{R}}}

\newcommand{\NN}{\ensuremath{\mathbb N}}
\newcommand{\FF}{\ensuremath{\mathcal F}}
\newcommand{\ind}{\ensuremath{\mathds 1}}

\newcommand{\Gr}{\ensuremath{\operatorname{Gr}}}

\newcommand{\var}{\ensuremath{\operatorname{var}}}

\newcommand{\expo}{\ensuremath{\operatorname{exp}}}

\newcommand{\BR}{\ensuremath{\operatorname{BR}}}

\newcommand{\Argmax}{\operatornamewithlimits{argmax}}

\newcommand{\norm}[1]{\left\Vert #1 \right\Vert} 

\numberwithin{equation}{section}
\newtheorem{theorem}{Theorem}[section]
\newtheorem{definition}[theorem]{Definition}
\newtheorem{proposition}[theorem]{Proposition}
\newtheorem{lemma}[theorem]{Lemma}

\theoremstyle{definition}
\newtheorem{remark}[theorem]{Remark}

\begin{document}
\date{}
\maketitle
\date{}
\begin{abstract}
\noindent Consider a 2-player normal-form game repeated over time. We introduce an adaptive learning procedure, where the players only observe their own realized payoff at each stage. We assume that agents do not know their own payoff function, and have no information on the other player. Furthermore, we assume that they have restrictions on their own action set such that, at each stage, their choice is limited to a subset of their action set. We prove that the empirical distributions of play converge to the set of Nash equilibria for zero-sum and potential games, and games where one player has two actions. 
\end{abstract}

\section{Introduction}
\noindent  First introduced by Brown~\cite{brown51} to compute the value of zero-sum games, {\it fictitious play} is one of the most intensely studied and debated procedures in game theory. Consider an $N$-player normal form game which is repeated in discrete time. At each time, players compute a \textit{best response} to the opponent's empirical average play.

\noindent A major issue in fictitious play is identifying classes of games where the empirical frequencies of play converge to the set of  Nash equilibria of the underlying game. A large body of literature has been devoted to this question.  Convergence for 2-player zero-sum  games was obtained by Robinson~\cite{robinson51} and for general (non-degenerate) $2 \times 2$ games by Miyasawa~\cite{Miy61}. Monderer and Shapley \cite{MonSha96b} proved the same result for potential games, and Berger~\cite{berger05} for $2$-player games where one of the players has only two actions. Recently, a large proportion of these results have been re-explored using the stochastic approximation theory (see for example, Bena\"im~\cite{benaim99}, Benveniste et al.~\cite{bmp90}, Kushner and Yin~\cite{ky03}), where the asymptotic behavior of the fictitious play procedure  can be analyzed through related dynamics. For instance, Hofbauer and Sorin~\cite{hs06} obtain more general convergence results for zero-sum games, while Bena\"{i}m, Hofbauer and Sorin \cite{bhs05} extend Monderer and Shapley's result to a general class of potential games, with nonlinear payoff functions on compact convex action sets.

\noindent Most of these convergence properties also hold for \emph{smooth fictitious play}, introduced by Fudenberg and Kreps~\cite{fk93}, (see also \cite{fl98}), where agents use a fictitious play strategy in a game where payoff functions are perturbed by random variables, in the spirit of Harsanyi~\cite{Har73}. For this adaptive procedure, convergence holds in $2 \times 2$ games (see \cite{bh99a}), zero-sum, potential games (see \cite{hs02}), and supermodular games (see \cite{bf12}).

\noindent As defined above, in fictitious play or smooth fictitious play,
players compute best responses to their opponents' empirical
frequencies of play. Three main assumptions are made here: (i) each player knows the structure of the game, i.e. she knows her own payoff function; (ii)
each player is informed of the action selected by her opponents at each stage; thus
she can compute the empirical frequencies; (iii) each player is
allowed to choose any action at each time, so that she can actually play a
best response.  

\noindent The next question is usually, what happens if assumptions (i) and (ii) are relaxed. One approach is to assume that the agents observe only their realized payoff at each stage. This is the minimal information framework of the so-called  {\em reinforcement learning} procedures (see \cite{bs97, er98} for pioneer work on this topic). Most work in this direction proceeds as follows: $a)$ construct a sequence of mixed strategies which are updated taking into account the payoff they receive (which is the only information agents have access to) and $b)$ study the convergence (or non-convergence) of this sequence. It is supposed that players are given a rule of behavior (a {\em decision rule}) which depends on a {\em state variable} constructed by means of the aggregate information they gather and their own history of play. 

\noindent It is noteworthy that most of the decision rules considered in the literature are {\it stationary} in the sense that they are defined through a time-independent function of the state variable. This kind of rule has proved useful in the analysis of simple cases (e.g. $2 \times 2$ games \cite{posch97}), 2-players games with positive payoff \cite{bs97, beggs05, hopkins02, hp05} or in establishing convergence to perturbed equilibria in 2-player games \cite{lc05} or multiplayer games \cite{cms10}. An example of a non-homogeneous (time-dependent) decision rule is proposed by Leslie and Collins~\cite{lc06} where, via stochastic approximation techniques, convergence of mixed actions is shown for zero-sum games and multiplayer potential games.  Another interesting example that implements a non-homogeneous decision rule is proposed by Hart and Mas-Colell \cite{hmc01}. Using techniques based on consistent procedures (see Hart and Mas-Colell~\cite{hmc00}), the authors show that, for any game, the joint empirical frequency of play converges to the set of correlated equilibria. To our knowledge, this is the only reinforcement learning procedure that uses a decision rule depending explicitly on the last action played (i.e. it is {\em Markovian}). However, in all the examples described above, assumption (iii) holds; in other words, players can use any action at any time. 

\noindent A different idea, that of releasing assumption (iii), comes from
Bena\"im and Raimond~\cite{br10}, who introduced the \textit{Markovian fictitious play} (MFP) procedure, where players have restrictions on their action set, due to limited computational capacity or even to physical restrictions. Players know the structure of the game and, at each time, they are informed of opponents' actions, as in the fictitious play framework. Under the appropriate conditions regarding payers' ability to explore their action set, it is shown that this adaptive procedure converges to  Nash equilibria for zero-sum and potential games.
 
\noindent Here, we drop all three assumptions (i), (ii) and (iii). The main novelty of this work is that we construct a sophisticated, non-stationary learning procedure in $2$-player games with minimal information and restrictions on players' action sets. We assume that players do not anticipate opponents' behavior  and  that they have no information on the structure of the game (in particular, they do not know their own payoff function) nor on opponents' actions at each stage. This means that the only information allowing agents to react to the environment is their past realized payoffs; the adaptive procedure presented in this work thus belongs to the class of reinforcement learning algorithms. In addition (and in the spirit of the (MFP) procedure), we suppose that at each stage the agents are restricted to a subset of their action set, which depends on the action they chose at the previous stage.  The decision rule we implement is fully explicit, and it is easy for each agent to compute the mixed strategy which dictates her next action. She actually chooses an action through a non-homogeneous Markovian rule which depends on a meaningful state variable.  

\noindent One of the main differences between this procedure and standard reinforcement learning is that the  sequence of mixed strategies is no longer a natural choice of state variable. Indeed, the set of mixed strategies available to a given agent at time $n+1$ depends on the action he chose at time $n$. As a consequence, it is unrealistic to expect good asymptotic behavior from the sequence of mixed strategies, and  we turn our attention to the sequence of empirical moves. Our main finding is that the empirical frequencies of play converge to Nash equilibria in zero-sum and potential games, including convergence of the average scored payoffs. We also show convergence in the case where at least one player has only two actions. 

\noindent This paper is organized as follows.  In Section~\ref{sec:model} we describe the setting and present our model, along with our main result. Section~\ref{sec:preliminaries} introduces the general framework in which we analyze our procedure. The related Markovian fictitious play procedure is also presented, to help the reader better grasp our adaptive procedure. Section~\ref{proofs} gives the proof of our main result,  presented as an extended sketch, while the remaining results and technical comments are left to the Appendix. 

\section{The Model}\label{sec:model}

\subsection{Setting}

Let $\mathcal{G} = (N,(S^i)_{i \in N}, (G^i)_{i \in N})$ be a given finite normal form game and $S = \prod_i S^i$ be the set of action profiles. We call $\Delta(S^i)$ the mixed action set, i.e
\begin{equation*}
\Delta(S^i)= \left \{ \sigma^i \in \RR^{|S^i|} \, :\, \sum_{s^i \in S^i} \sigma^i(s^i) =1, \, \sigma^i(s^i) \geq 0, \,\forall s^i \in S^i \right\},
\end{equation*}
and $\Delta = \prod_i \Delta(S^i)$.  More generally, given a finite
set $S$, $\Delta(S)$ denotes the set of probability distributions over
$S$.  

\noindent In the whole paper, for any agent $i$, we denote $\delta_{s^i}$ the pure strategy $i$ seen as an element of $\Delta(S^i)$. As usual, we use the notation $-i$ to exclude player $i$, namely $S^{-i}$ denotes the set  $\prod_{j \neq i} S^j$ and $\Delta^{-i}$ the set $\prod_{j \neq i} \Delta(S^i)$. 

\begin{definition}
The Best-Response correspondence for player $i \in N$, $\BR^i: \Delta^{-i} \rightrightarrows \Delta(S^i)$, is defined as 
\begin{equation*}
  \BR^i(\sigma^{-i})= \Argmax \limits_{\sigma^i \in \Delta(S^i)} G^i(\sigma^i, \sigma^{-i}).\,
\end{equation*}
 for any $\sigma^{-i}\in \Delta^{-i}$. The Best-Response correspondence $\BR: \Delta \rightrightarrows \Delta$ is given by 
 \begin{equation*}
   \BR(\sigma)= \prod_{i \in N} \BR^i(\sigma^{-i}) ,
 \end{equation*}
for all $\sigma \in \Delta$.
\end{definition}

\noindent Recall that a Nash equilibrium of the game $\mathcal G$ is a fixed point of the set-valued map $\BR$, namely a mixed action profile $\sigma^* \in \Delta$ such that $\sigma^* \in \BR(\sigma^*)$. 

\subsection{Payoff-based Markovian procedure} \label{sec:MFP}

We consider a situation where the game $\mathcal G$ described above is repeated in discrete time. Let $s_n^i \in S^i$ be the action played by player $i$ at time $n$. We assume that players do not know the game that they are playing, i.e. they know  neither their own payoff functions nor opponents'.  Also we assume that the information that a player can gather at any stage of the game is given by her payoff, i.e. at each time $n$ each player $i \in N$ is informed of 
\begin{equation*}
g_n^i=G^i(s_n^1,s_n^2,...,s_n^N).
\end{equation*}
Players are not able to observe opponents' actions. 

\noindent In this framework, a {\em reinforcement learning} procedure can be defined in the following manner.
Let us assume that, at the end of stage $n \in \NN$, player $i$ has constructed a {\em state variable} $X_n^i$. Then
\begin{itemize}
\item[$(a)$] at stage $n+1$, player $i$ selects a mixed strategy $\sigma_n^i$ according to a {\em decision rule}, which can depend on state variable $X_n^i$ the time $n$.
\item[$(b)$] Player $i$'s action $s_{n+1}^i$ is randomly drawn according to $\sigma_n^i$.
\item[$(c)$]  She only observes $g_{n+1}^i$, as a consequence of the realized action profile $(s_{n+1}^1,\ldots,s_{n+1}^N)$.
\item[$(d)$] Finally, this observation allows her to update her state variable to $X_{n+1}^i$ through an {\em updating rule}, which can depend on  observation $g_{n+1}^i$,  state variable $X_n^i$, and time $n$.
\end{itemize}

\noindent  In this work we assume that, in addition, players have restrictions on their action set. This idea was introduced by Bena\"{\i}m and Raimond ~\cite{br10} through the definition of the (MFP) procedure (see Section~\ref{sec:mark-fict-play} for details). Suppose that, when an agent $i$  plays a pure strategy $s \in S^i$ at stage $n \in \NN$, her available actions at stage $n+1$ are reduced to a subset of $S^i$. This can be due to physical restrictions, computational limitations or a large number of available actions.  The subset of actions available to player $i$ depends on her last action and is defined through a stochastic \emph{exploration matrix} $M_0^i \in \RR^{|S^i|}$. In other words, if at stage $n$ player i plays $s \in S^i$,  she can switch to action  $r \neq s$ at stage $n+1$ if and only if $M_0^i(s,r)>0$. 

\noindent  The matrix $M_0^i$ is assumed to be irreducible and reversible with respect to its unique invariant measure $\pi^i_0$, i.e.
\begin{equation*} \label{eq:pi_0}
 \pi_0^i(s)M_0^i(s,r)= \pi_0^i(r)M_0^i(r,s),
\end{equation*}
\noindent for every $s,r \in S^i$. This assumption guarantees that agents have access to any of their actions.
\begin{remark}
Recall that a stochastic matrix $M$ over a finite set $S$ is said to be irreducible if it has a unique recurrent class which is given by $S$.
\end{remark}

\noindent For $\beta>0$ and a vector $R \in \RR^{|S^i|}$, we define the stochastic matrix $M^i[\beta, R]$ as
\begin{equation} \label{def_M}
  M^i[\beta,R](s,r)= 
\begin{cases} 
M_0^i(s,r) \expo(-\beta  |R(s)- R(r)|_+) &s \neq r \\
 1 - \sum \limits_{s' \neq s} M^i[\beta,R](s,s') & s =r, \\ 
\end{cases}
\end{equation}
where, for a number $a \in \RR$, $|a|_+= \max\{a,0 \}$. 

\noindent From the irreducibility of the exploration matrix $M_0^i$, we have that $M^i[\beta,R]$ is also irreducible and its unique invariant measure is given by
\begin{equation}
\pi^i[\beta,R](s)= \frac{\pi_0^i(s)\expo(\beta R(s))}{\sum \limits_{r \in S^i}\pi_0^i(r)\expo(\beta R(r))}, \label{mesinv}
\end{equation}
for any $\beta >0$, $R \in \RR^{|S^i|}$, and $s \in S^i$.

\noindent Let $(\beta_n^i)_n$ be a deterministic sequence and let $\mathcal F_n$ be the sigma algebra generated by the history of play up to time $n$.
We suppose that, at the end of stage $n$, player $i$ has a state variable $R_n^i \in \mathbb{R}^{|S^i|}$. Let $M_n^i= M^i[\beta_n^i, R_n^i]$ and $\pi_n^i= \pi_n^i[\beta_n^i, R_n^i]$.  Player $i$ selects her action at time $n+1$ through the
following \emph{choice rule}:
\begin{equation}\label{strategy}  \tag{CR}
\begin{aligned}
 \sigma_n^i(s)=\mathbb P(s_{n+1}^i=s \, \vert \, \mathcal F_n) &=   M_n^i(s_n^i,s),\\
&= 
\begin{cases}
 M_0^i(s_n^i,s) \expo(-\beta_n^i |R_n^i(s_n^i) - R_n^i(s)|_+) &  s \neq s_n^i \\ 
1 - \sum \limits_{s' \neq s}  M_n^i(s_n^i,s') & s =s_n^i. 
 \end{cases}
 \end{aligned}
\end{equation}
for every $s \in S^i$. As we will see, variable $R_n^i$ will
be defined so as to be an estimator of the
time-average payoff vector.

\noindent At time $n+1$, player $i$ observes her realized payoff $g_{n+1}^i$.  The \textit{updating rule} chosen by player $i$ is defined as follows. Agent $i$ updates the vector $R_n^i \in \mathbb{R}^{|S^i|}$, only on the component associated to the action selected at stage $n$. For every action $s \in S^i$,  
\begin{equation} \label{eq:ur} \tag{UR}
 R_{n+1}^i(s)  = R_n^i(s) +\gamma_{n+1}^i(s) \left ( g_{n+1}^i - R^i_n(s)\right )\ind_{\{s_{n+1}^1=s\}},\\
\end{equation}
where, \[ \gamma^i_{n+1}(s) = \min \left \{1\,, \frac{1}{(n+1)\pi^i_{n}(s)}\right\},\]
and $\ind_E$ is the indicator of the event $E$.

\begin{remark}
Strictly speaking, the state variable is of the form $X_n^i=(R_n^i,s_n^i)$, since the choice rule \eqref{strategy} is Markovian. We use this interpretation for the sake of simplicity.
\end{remark}

\noindent Note that the step size $\gamma_{n+1}^i(s)$ depends only on $\pi_0^i$,  $\beta_n^i$ and $R_n^i$. Also, as we will see later on,$\gamma_n^i(s) = 1/(n\pi^i_{n-1}(s))$ for sufficiently large $n$ (\emph{c.f.} Section~\ref{lm:R}). 

\noindent While choosing this step size might appear surprising, we
believe that it is actually very natural, as it takes advantage of the
fact that the invariant distribution $\pi_n^i$ is known by player $i$. To put it another way: a natural candidate for step size $\gamma_n^i(s)$ in \eqref{eq:ur} is $\gamma_n^i(s) = 1/\theta_n^i(s)$, where $\theta_n^i(s)$ is equal to the number of times agent $i$ actually played action $s$ during the $n$ first steps. If the Markov process was homogeneous and ergodic, with invariant measure $\pi^i$, then the expected value of $\theta_n^i$ would be exactly $n \pi^i(s)$. 

\noindent Consequently, our stochastic approximation scheme
\eqref{eq:ur} can be interpreted as follows. Assume that, at time
$n+1$, action $s$ is played by agent $i$. Then $R_{n+1}(s)$ is
updated by taking a convex combination of $R_n(s)$ and of the
realized payoff playing $s$ at time $n+1$; additionally  the weight
that is put on the  realized payoff is inversely proportional to
the number of times this action \emph{should} have been played
(and not the number of times it has \emph{actually} been played).

\noindent Let us denote by $(v_n^i)_n$ the sequence of empirical distribution of moves of agent $i$, i.e.
\[v^i_n = \frac{1}{n} \sum_{m=1}^n \delta_{s_m^i},\] 
and $v_n = (v_n^i)_{i \in N} \in \Delta$.

\noindent  Note that, given the physical restrictions on the action set, one cannot expect convergence results on the mixed actions of players $\sigma_n^i$. Therefore, the empirical frequencies of play become the natural focus of our analysis.

\begin{definition}
\noindent We call \emph{Payoff-based Markovian procedure} the adaptive process where, for any $i \in N$, agent $i$ plays according to the choice rule \eqref{strategy}, and updates $R_n^i$ through the updating rule \eqref{eq:ur}.
\end{definition}

\subsection{Main result}\label{sec:mainresult}
\noindent In the case of a 2-player game, we introduce our major assumption on the positive sequence $(\beta_n^i)_n$. Let us assume that, for $i \in \{1,2\}$,
\begin{equation}
\begin{aligned}
\text{(i)  } & \beta_n^i \longrightarrow +\infty, \\
\text{(ii)  }&\beta_n^i \leq  A_n^i \ln(n),  \textup{ where }  A_n^i \longrightarrow 0.  \\
\end{aligned} \label{hyp_param} \tag{$H$}
\end{equation}

\noindent Let us denote by $\overline{g}_n^i$ the average payoff obtained by player $i$, i.e.
\begin{equation} \label{p_moyen}
\overline{g}_n^i = \frac{1}{n} \sum_{m=1}^n G^i(s_m^1,s_m^2),
\end{equation}
and $\overline{g}_n = (\overline{g}_n^1, \overline{g}_n^2)$.

\noindent For a sequence $(z_n)_n$, we call $\mathcal{L}((z_n)_n)$ its limit set , i.e.
\[\mathcal{L}((z_n)_n) = \big \{ z : \;  \mbox{ there exists a subsequence }  \,  (z_{n_k})_k \; \mbox{such that }  \, \lim \nolimits_{k \to +\infty} z_{n_k} = z\big\}.\]
We say that the sequence $(z_n)_n$ converges to a set $A$ if $\mathcal{L}((z_n)_n) \subseteq A$.

\noindent  Recall that $\mathcal G$ is a potential game with potential $\Phi$ if, for all $i=1,2$, and $s^{-i} \in S^{-i}$, we have $G^i(s^i,s^{-i}) - G^i(t^i,s^{-i}) = \Phi(s^i,s^{-i}) - \Phi(t^i,s^{-i})$,  for all $s^i,t^i \in S^i$.

\noindent Our main result is the following.
\begin{theorem} \label{th:main}
Under assumption \eqref{hyp_param}, the Payoff-based Markovian  procedure enjoys the following properties:
\begin{enumerate}
\item[\textup{(a)}] In a zero-sum game, $(v_n^1, v_n^2)_n$ converges almost surely to the set of Nash equilibria and the average payoff $(\overline g_n^1)_n$ converges almost surely to the value of the game. 
\item[\textup{(b)}]  In a potential game with potential $\Phi$,  $(v_n^1, v_n^2)_n$ converges almost surely to a connected subset of the set of Nash equilibria on which $\Phi$ is constant, and $\frac{1}{n} \sum_{m=1}^n \Phi(s_m^1,s_m^2)$ converges to this constant.  

In the particular case $G^1 = G^2$, then $(v_n^1, v_n^2)_n$ converges almost surely to a connected subset of the set of Nash equilibria on which $G^1$ is constant; moreover $(\overline g_n^1)_n$ converges almost surely to this constant. 
\item[\textup{(c)}]  If either $|S^1|=2$ or $|S^2| =2$, then $(v_n^1, v_n^2)_n$ converges almost surely to the set of Nash equilibria.
\end{enumerate}
\end{theorem}

\noindent In fact, we prove a more general result. We establish a relationship between the limit set of the sequence $(v_n^1, v_n^2)_n$ and the attractors of the well-known Best-Response dynamics \cite{gm91}

\begin{equation} \label{BR}
 \dot v \in  -v + \BR(v). \tag{BRD}
\end{equation}
See Section~\ref{proofs} (Theorem \ref{main2}) for details.

\paragraph{Comments on the result} For potential games, in the general case, the payoff of a given player is not necessarily constant on the limit set of $(v_n)_n$. However, the potential almost surely is.

\noindent Consider the game $\mathcal{G}$, with payoff function $G$ and potential $\Phi$:
\newcommand{\mc}[3]{\multicolumn{#1}{#2}{#3}}
\begin{center}
  \begin{equation} \tag{$\mathcal{G}$}
    \hfill
G=
\begin{tabular}{cccc}
 & $a$ & $b$ & $c$ \\\cline{2-4}
\mc{1}{l|}{$A$} & \mc{1}{l|}{1,1} & \mc{1}{l|}{9,0} & \mc{1}{l|}{1,0}\\\cline{2-4}
\mc{1}{l|}{$B$} & \mc{1}{l|}{0,9} & \mc{1}{l|}{6,6} & \mc{1}{l|}{0,8}\\\cline{2-4}
\mc{1}{l|}{$C$} & \mc{1}{l|}{0,1} & \mc{1}{l|}{9,0} & \mc{1}{l|}{2,2}\\\cline{2-4}
\end{tabular}\qquad\text{       and     } \qquad 
\Phi= \begin{tabular}{cccc}
 & $a$ & $b$ & $c$ \\\cline{2-4}
\mc{1}{l|}{$A$} & \mc{1}{l|}{4} & \mc{1}{l|}{3} & \mc{1}{l|}{3}\\\cline{2-4}
\mc{1}{l|}{$B$} & \mc{1}{l|}{3} & \mc{1}{l|}{0} & \mc{1}{l|}{2}\\\cline{2-4}
\mc{1}{l|}{$C$} & \mc{1}{l|}{3} & \mc{1}{l|}{2} & \mc{1}{l|}{4}\\\cline{2-4}
\end{tabular}\hfill
\end{equation}
\end{center}
There is a mixed Nash equilibrium, and two strict Nash equilibria $(A,a)$ and $(C,c)$, with same potential value (equal to $4$). However, $\mathbb{P} \left[ \mathcal{L}((v_n)_n) = \{(A,a),(C,c)\}\right] = 0$, because this set is not connected.

\noindent Now consider the following  modified version $\mathcal{G}'$:
\begin{center}
  \begin{equation}\label{p_game} \tag{$\mathcal{G}'$}
    \hfill
G'=
\begin{tabular}{cccc}
 & $a$ & $b$ & $c$ \\\cline{2-4}
\mc{1}{l|}{$A$} & \mc{1}{l|}{1,1} & \mc{1}{l|}{9,0} & \mc{1}{l|}{1,0}\\\cline{2-4}
\mc{1}{l|}{$B$} & \mc{1}{l|}{0,9} & \mc{1}{l|}{6,6} & \mc{1}{l|}{0,8}\\\cline{2-4}
\mc{1}{l|}{$C$} & \mc{1}{l|}{1,2} & \mc{1}{l|}{8,0} & \mc{1}{l|}{2,2}\\\cline{2-4}
\end{tabular}\qquad\text{       and     } \qquad 
\Phi'= \begin{tabular}{cccc}
 & $a$ & $b$ & $c$ \\\cline{2-4}
\mc{1}{l|}{$A$} & \mc{1}{l|}{4} & \mc{1}{l|}{3} & \mc{1}{l|}{3}\\\cline{2-4}
\mc{1}{l|}{$B$} & \mc{1}{l|}{3} & \mc{1}{l|}{0} & \mc{1}{l|}{2}\\\cline{2-4}
\mc{1}{l|}{$C$} & \mc{1}{l|}{4} & \mc{1}{l|}{2} & \mc{1}{l|}{4}\\\cline{2-4}
\end{tabular}\hfill
\end{equation}
\end{center}
Here we see that the set of Nash equilibria is connected and equal to 
\begin{equation}
  \label{eq:NE}
NE = \left\{((x,0,1-x),a), \; x \in [0,1] \right\} \cup \left\{(C,(y,0,1-y)), \; y \in [0,1] \right\}.
\end{equation}
\noindent Consequently, there is no reason to rule out the possibility that the limit set of $(v_n)_n$ is equal to the whole set of Nash equilibria. Therefore the payoff is not necessarily constant on $\mathcal{L}((v_n)_n)$.
\paragraph{Comments on the assumptions.}
 Condition~\eqref{hyp_param} assumes that the sequence $\beta_n^i$ increases to infinity as $o(\ln(n))$. This assumption is necessary due to the informational constraints on players. For instance, it is not possible to know {\em a priori} how far the variables $R_0^i$ are from the set of feasible payoffs.

\noindent As we will see later on, in the Markovian fictitious play procedure,  sequence $\beta_n^i$ is supposed to grow more slowly than $A^i \ln(n)$, where $A^i$ is smaller than a quantity which is related to the \emph{energy barrier} of the payoff matrix $G^i$ (see \cite{br10} for details). This quantity is in turn related to the {\em freezing schedule} of the simulated annealing algorithm (see for example  \cite{HolleyStroock88, hajek88} and references therein).

\noindent We believe it is worth reformulating our result in this spirit. However, this requires players to have more information about the game. For each $i \in \{1,2\}$,  suppose that the initial state variable $R_0^i$ belongs to the set of feasible payoffs. Also, let us define the quantity
\begin{equation*}
  \omega^i=\max_{s \in S^i}\max_{s^{-i}, r^{-i} \in S^{-i}}|G^i(s, s^{-i}) -G^i(s, r^{-i})|,
\end{equation*}
and let us assume that player $i$ can choose a positive constant $A^i$ such that
\begin{equation}
\begin{aligned}
\text{(i)} & \beta_n^i \longrightarrow +\infty, \\
\text{(ii)}&\beta_n^i \leq  A^i \ln(n),  \textup{ where } 2A^i\omega^i < 1.  \\
\end{aligned} \label{hyp_param'} \tag{$H'$}
\end{equation}
Then, we have the following version of our main result.
\begin{theorem}
  Under Assumption~\eqref{hyp_param'}, conclusions of Theorem~\ref{th:main} hold.
\end{theorem}
\noindent The proof of this result runs along the same lines as the proof of Theorem~\ref{th:main}, and is therefore omitted.

\subsection{Examples}
\noindent The following simple examples show the scope of our main result. In every case presented in this section, we performed a maximum of $5 \times 10^{5}$ iterations. 
\paragraph{Blind-restricted RSP.}  Consider the  Rock-Scissor-Paper game defined by the payoff matrix $G^1$: 
\vspace{-.1cm}
  \begin{equation}
    \label{rsp} \tag{$RSP$}
\begin{tabular}{cccc}
 & $R$ & $S$ & $P$ \\\cline{2-4}
\mc{1}{l|}{$R$} & \mc{1}{l|}{\phantom{-}0} & \mc{1}{l|}{\phantom{-}1} & \mc{1}{l|}{-1}\\\cline{2-4}
\mc{1}{l|}{$S$} & \mc{1}{l|}{-1} & \mc{1}{l|}{\phantom{-}0} & \mc{1}{l|}{\phantom{-}1}\\\cline{2-4}
\mc{1}{l|}{$P$} & \mc{1}{l|}{\phantom{-}1} & \mc{1}{l|}{-1} & \mc{1}{l|}{\phantom{-}0}\\\cline{2-4}
\end{tabular}
\end{equation}

\noindent Then the optimal strategies are given by $((1/3,1/3,1/3),(1/3,1/3,1/3)) \in \Delta$ and the value of the game is $0$. Assume that players' exploration matrices and their invariant measures are given by 
\begin{equation}
  \label{eq:exploration}
  M_0^1 = M_0^2 = %
\caption{\small{Graph representing players' restrictions for the \eqref{rsp} game.}}
\label{fig1}
\end{center}
\end{figure}
Figure~\ref{fig1} means that if a player's action is \emph{Rock} at some time, she cannot select \emph{Paper} immediately afterwards, and inversely. In Figure~\ref{rsp_figure}, we present a realization of $(v_n)_n$, as well as $(g_n^1)_n$.
\begin{figure}[h!]
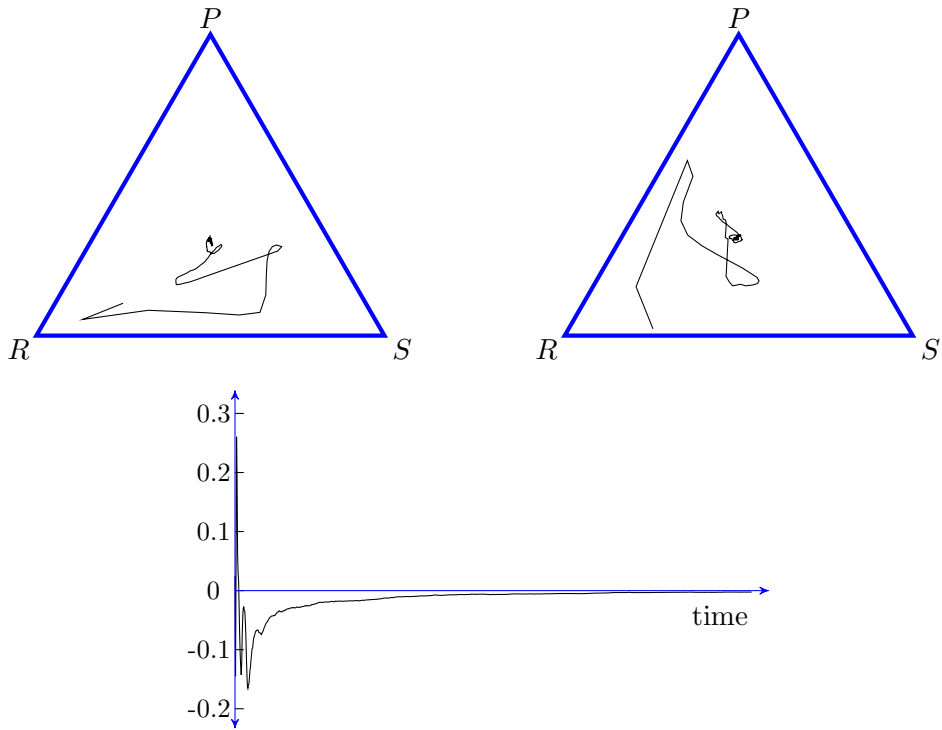

\begin{center}
   \subfloat{
%
}
\end{center}
\caption[lineheight]{\small{At the top, a realization of $v_n$ ($v_n^1$ on the left, $v_n^2$ on the right). At the bottom, $g_n^1$.}}
\label{rsp_figure}
\end{figure}
 
\paragraph{$3\times 3$ Potential game.} Consider the potential game with payoff matrix $G'$ and potential $\Phi'$ (see \eqref{p_game}).  We assume that players' exploration matrices are also given by \eqref{eq:exploration}. Therefore the graph representing the restriction of players is given by Figure~\ref{fig1}, if $R,S$ and $P$ are replaced by $A,B$ and $C$, respectively.

\noindent Figure~\ref{potential_figure} shows a realization of our procedure for the game~\eqref{p_game}. On the left, we plot the evolution of $v_n^1$. On the right, we present the corresponding trajectory of $\overline{\Phi}'_n = \frac{1}{n} \sum_{m=1}^n \Phi'(s_m^1,s_m^2)$, the average value of the potential $\Phi'$ along the realization of $(s_n^1,s_n^2)_n$. Note that our results do not stipulate that $(v_n)_n$ converges, and that our simulation tend towards non-convergence of $v_n^1$. 
We choose not to display $v_n^2$ here (which seems to converge to the action $a$).

\begin{figure}[h!]
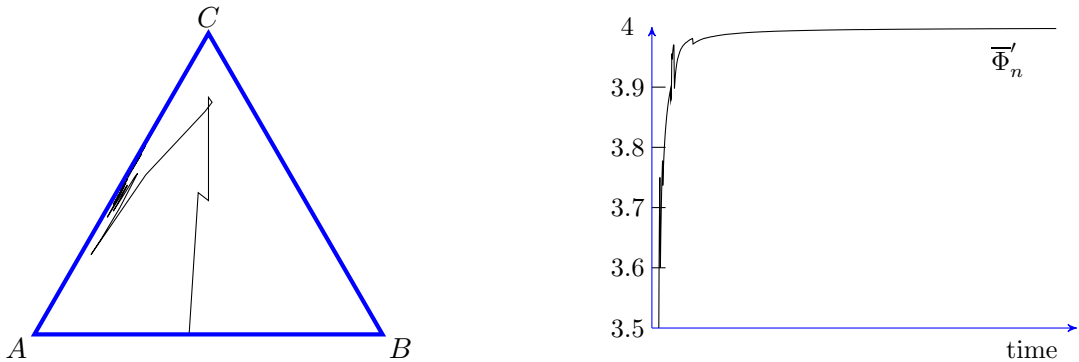

\begin{center}
   \subfloat{
  
 } 
\end{center}  
\caption[lineheight]{\small{Simulation results for the potential game \eqref{p_game}.}}
\label{potential_figure}
\end{figure}

\paragraph{$5\times 5$ Identical interests  game. }
Consider the game with identical interests where both players have $5$ actions and the common payoff matrix is given by
\begin{center}
  \begin{equation}\label{coordination} \tag{$C$}
    \hfill
\begin{tabular}{cccccc}
 & $A$ & $B$ & $C$ & $D$ & $E$ \\\cline{2-6}
\mc{1}{l|}{$A$} & \mc{1}{l|}{2} & \mc{1}{l|}{0} & \mc{1}{l|}{0}& \mc{1}{l|}{0} & \mc{1}{l|}{0}\\\cline{2-6}
\mc{1}{l|}{$B$} & \mc{1}{l|}{0} & \mc{1}{l|}{1} & \mc{1}{l|}{0}&  \mc{1}{l|}{0} & \mc{1}{l|}{0}\\\cline{2-6}
\mc{1}{l|}{$C$} & \mc{1}{l|}{0} &\mc{1}{l|}{0} & \mc{1}{l|}{0}&\mc{1}{l|}{0} &\mc{1}{l|}{0}\\  \cline{2-6}
\mc{1}{l|}{$D$} & \mc{1}{l|}{0} &\mc{1}{l|}{0} & \mc{1}{l|}{0}&\mc{1}{l|}{1} &\mc{1}{l|}{0}\\  \cline{2-6}
\mc{1}{l|}{$E$} & \mc{1}{l|}{0} &\mc{1}{l|}{0} & \mc{1}{l|}{0}&\mc{1}{l|}{0} &\mc{1}{l|}{2}\\  \cline{2-6}
\end{tabular}\hfill
  \end{equation}
\end{center}

\noindent Assume that players' exploration matrices are 

\[M_0^1 = M_0^2 = %
\caption{\small{Graph representing players' restrictions for the game \eqref{coordination}.}}
\end{figure}

\noindent  Note that, even if the center action $C$ is bad for both players, the restrictions force them to play $C$ every time they switch to another action. 

\noindent In Figure~\ref{fig_coordination}, on the left, we present a realization where $(v_n^1,v_n^2)$ converges to the NE $(B,B)$. On the right, a trajectory where $(v_n^1,v_n^2)$ converges to the NE $(E,E)$ is displayed. Note that, in both cases, the average realized payoff $\overline g_n$ converges to the payoff of the corresponding equilibrium. For simplicity, we only plot the component that converges to one for the first player. This is consistent with the recent finding that the four strict Nash equilibria have a positive probability of being the limit of the random process $(v_n)_n$ (see \cite{FauRot10} for details).
\begin{figure}[h!]
\begin{center}
   \subfloat{
\begin{tikzpicture}[scale=0.5]
\draw (1.196,0.368)--(1.376,0.368);
\node at (0.55,0.368) {\small 0.4};
\draw (1.196,3.039)--(1.376,3.039);
\node at (0.55,3.039) {\small 0.6};
\draw (1.196,5.710)--(1.376,5.710);
\node at (0.55,5.710) {\small 0.8};
\node at (0.55,8.381) {\small 1};

\node at (11.3,-0.2) {\small{time}};
\draw[->,>=stealth', color=blue](1.196,0.368)-- (1.196,8.8);
\draw[->,>=stealth', color=blue](1.196,0.368)--(12.5,0.368);
\node at (10.663,7.5) {\small {$v_n^1(B)$}};
\draw (1.196,5.910)--(1.205,6.556)--(1.214,6.611)--(1.223,6.431)--(1.232,6.578)%
  --(1.241,6.664)--(1.250,6.879)--(1.259,6.956)--(1.268,7.099)--(1.277,7.094)--(1.286,6.945)%
  --(1.295,7.056)--(1.304,7.084)--(1.313,7.054)--(1.322,7.137)--(1.330,7.210)--(1.339,7.276)%
  --(1.348,7.334)--(1.357,7.326)--(1.366,7.376)--(1.375,7.422)--(1.384,7.464)--(1.393,7.502)%
  --(1.402,7.537)--(1.411,7.523)--(1.420,7.555)--(1.429,7.584)--(1.438,7.612)--(1.447,7.638)%
  --(1.456,7.662)--(1.465,7.684)--(1.474,7.705)--(1.483,7.725)--(1.492,7.740)--(1.501,7.758)%
  --(1.510,7.775)--(1.519,7.791)--(1.528,7.806)--(1.537,7.820)--(1.546,7.834)--(1.555,7.847)%
  --(1.564,7.859)--(1.573,7.871)--(1.582,7.882)--(1.591,7.893)--(1.599,7.904)--(1.608,7.914)%
  --(1.617,7.923)--(1.626,7.932)--(1.635,7.941)--(1.644,7.950)--(1.653,7.938)--(1.662,7.946)%
  --(1.671,7.954)--(1.680,7.961)--(1.689,7.969)--(1.698,7.957)--(1.707,7.965)--(1.716,7.971)%
  --(1.725,7.978)--(1.734,7.985)--(1.743,7.991)--(1.752,7.997)--(1.761,8.003)--(1.770,8.009)%
  --(1.779,8.014)--(1.788,8.018)--(1.797,8.023)--(1.806,8.028)--(1.815,8.033)--(1.824,8.038)%
  --(1.833,8.043)--(1.842,8.047)--(1.851,8.052)--(1.860,8.056)--(1.868,8.060)--(1.877,8.064)%
  --(1.886,8.068)--(1.895,8.072)--(1.904,8.058)--(1.913,8.062)--(1.922,8.066)--(1.931,8.069)%
  --(1.940,8.073)--(1.949,8.077)--(1.958,8.080)--(1.967,8.084)--(1.976,8.087)--(1.985,8.090)%
  --(1.994,8.093)--(2.003,8.096)--(2.012,8.100)--(2.021,8.103)--(2.030,8.105)--(2.039,8.107)%
  --(2.048,8.110)--(2.057,8.113)--(2.066,8.115)--(2.075,8.118)--(2.084,8.121)--(2.093,8.123)%
  --(2.102,8.126)--(2.111,8.128)--(2.120,8.130)--(2.129,8.133)--(2.137,8.135)--(2.146,8.137)%
  --(2.155,8.140)--(2.164,8.142)--(2.173,8.144)--(2.182,8.146)--(2.191,8.148)--(2.200,8.150)%
  --(2.209,8.152)--(2.218,8.154)--(2.227,8.156)--(2.236,8.158)--(2.245,8.160)--(2.254,8.162)%
  --(2.263,8.164)--(2.272,8.165)--(2.281,8.136)--(2.290,8.132)--(2.299,8.134)--(2.308,8.136)%
  --(2.317,8.138)--(2.326,8.140)--(2.335,8.142)--(2.344,8.144)--(2.353,8.146)--(2.362,8.147)%
  --(2.371,8.149)--(2.380,8.151)--(2.389,8.152)--(2.398,8.137)--(2.406,8.139)--(2.415,8.141)%
  --(2.424,8.143)--(2.433,8.144)--(2.442,8.146)--(2.451,8.148)--(2.460,8.149)--(2.469,8.151)%
  --(2.478,8.153)--(2.487,8.154)--(2.496,8.156)--(2.505,8.157)--(2.514,8.159)--(2.523,8.160)%
  --(2.532,8.162)--(2.541,8.163)--(2.550,8.165)--(2.559,8.166)--(2.568,8.167)--(2.577,8.169)%
  --(2.586,8.170)--(2.595,8.171)--(2.604,8.173)--(2.613,8.174)--(2.622,8.175)--(2.631,8.177)%
  --(2.640,8.178)--(2.649,8.179)--(2.658,8.180)--(2.667,8.181)--(2.675,8.183)--(2.684,8.184)%
  --(2.693,8.185)--(2.702,8.186)--(2.711,8.187)--(2.720,8.188)--(2.729,8.190)--(2.738,8.191)%
  --(2.747,8.191)--(2.756,8.192)--(2.765,8.193)--(2.774,8.194)--(2.783,8.195)--(2.792,8.196)%
  --(2.801,8.197)--(2.810,8.198)--(2.819,8.199)--(2.828,8.200)--(2.837,8.201)--(2.846,8.202)%
  --(2.855,8.203)--(2.864,8.204)--(2.873,8.205)--(2.882,8.206)--(2.891,8.207)--(2.900,8.208)%
  --(2.909,8.209)--(2.918,8.210)--(2.927,8.194)--(2.936,8.195)--(2.944,8.196)--(2.953,8.197)%
  --(2.962,8.198)--(2.971,8.199)--(2.980,8.200)--(2.989,8.201)--(2.998,8.201)--(3.007,8.202)%
  --(3.016,8.203)--(3.025,8.204)--(3.034,8.205)--(3.043,8.206)--(3.052,8.207)--(3.061,8.207)%
  --(3.070,8.208)--(3.079,8.209)--(3.088,8.210)--(3.097,8.211)--(3.106,8.211)--(3.115,8.212)%
  --(3.124,8.213)--(3.133,8.214)--(3.142,8.215)--(3.151,8.215)--(3.160,8.216)--(3.169,8.217)%
  --(3.178,8.218)--(3.187,8.218)--(3.196,8.219)--(3.205,8.220)--(3.213,8.220)--(3.222,8.221)%
  --(3.231,8.222)--(3.240,8.222)--(3.249,8.223)--(3.258,8.224)--(3.267,8.225)--(3.276,8.225)%
  --(3.285,8.226)--(3.294,8.227)--(3.303,8.227)--(3.312,8.228)--(3.321,8.228)--(3.330,8.229)%
  --(3.339,8.230)--(3.348,8.230)--(3.357,8.231)--(3.366,8.232)--(3.375,8.232)--(3.384,8.233)%
  --(3.393,8.233)--(3.402,8.234)--(3.411,8.235)--(3.420,8.235)--(3.429,8.236)--(3.438,8.236)%
  --(3.447,8.237)--(3.456,8.237)--(3.465,8.238)--(3.474,8.239)--(3.482,8.239)--(3.491,8.240)%
  --(3.500,8.240)--(3.509,8.241)--(3.518,8.241)--(3.527,8.242)--(3.536,8.242)--(3.545,8.243)%
  --(3.554,8.243)--(3.563,8.244)--(3.572,8.244)--(3.581,8.245)--(3.590,8.245)--(3.599,8.246)%
  --(3.608,8.246)--(3.617,8.247)--(3.626,8.247)--(3.635,8.248)--(3.644,8.248)--(3.653,8.249)%
  --(3.662,8.249)--(3.671,8.250)--(3.680,8.250)--(3.689,8.251)--(3.698,8.251)--(3.707,8.252)%
  --(3.716,8.252)--(3.725,8.253)--(3.734,8.253)--(3.743,8.254)--(3.751,8.252)--(3.760,8.253)%
  --(3.769,8.253)--(3.778,8.253)--(3.787,8.254)--(3.796,8.254)--(3.805,8.255)--(3.814,8.255)%
  --(3.823,8.256)--(3.832,8.256)--(3.841,8.256)--(3.850,8.257)--(3.859,8.257)--(3.868,8.258)%
  --(3.877,8.258)--(3.886,8.259)--(3.895,8.259)--(3.904,8.259)--(3.913,8.260)--(3.922,8.260)%
  --(3.931,8.261)--(3.940,8.261)--(3.949,8.261)--(3.958,8.262)--(3.967,8.262)--(3.976,8.262)%
  --(3.985,8.263)--(3.994,8.263)--(4.003,8.264)--(4.012,8.264)--(4.020,8.264)--(4.029,8.265)%
  --(4.038,8.265)--(4.047,8.265)--(4.056,8.266)--(4.065,8.266)--(4.074,8.266)--(4.083,8.267)%
  --(4.092,8.267)--(4.101,8.268)--(4.110,8.268)--(4.119,8.268)--(4.128,8.269)--(4.137,8.269)%
  --(4.146,8.269)--(4.155,8.270)--(4.164,8.270)--(4.173,8.270)--(4.182,8.271)--(4.191,8.271)%
  --(4.200,8.271)--(4.209,8.272)--(4.218,8.272)--(4.227,8.272)--(4.236,8.273)--(4.245,8.273)%
  --(4.254,8.273)--(4.263,8.273)--(4.272,8.274)--(4.281,8.274)--(4.289,8.274)--(4.298,8.275)%
  --(4.307,8.275)--(4.316,8.275)--(4.325,8.276)--(4.334,8.276)--(4.343,8.276)--(4.352,8.276)%
  --(4.361,8.277)--(4.370,8.277)--(4.379,8.277)--(4.388,8.278)--(4.397,8.277)--(4.406,8.277)%
  --(4.415,8.278)--(4.424,8.278)--(4.433,8.278)--(4.442,8.279)--(4.451,8.279)--(4.460,8.279)%
  --(4.469,8.279)--(4.478,8.280)--(4.487,8.280)--(4.496,8.280)--(4.505,8.281)--(4.514,8.281)%
  --(4.523,8.281)--(4.532,8.281)--(4.541,8.282)--(4.550,8.282)--(4.558,8.282)--(4.567,8.282)%
  --(4.576,8.283)--(4.585,8.283)--(4.594,8.283)--(4.603,8.283)--(4.612,8.284)--(4.621,8.284)%
  --(4.630,8.284)--(4.639,8.284)--(4.648,8.285)--(4.657,8.285)--(4.666,8.285)--(4.675,8.285)%
  --(4.684,8.286)--(4.693,8.286)--(4.702,8.286)--(4.711,8.286)--(4.720,8.287)--(4.729,8.287)%
  --(4.738,8.287)--(4.747,8.287)--(4.756,8.288)--(4.765,8.288)--(4.774,8.288)--(4.783,8.288)%
  --(4.792,8.289)--(4.801,8.289)--(4.810,8.289)--(4.819,8.289)--(4.827,8.289)--(4.836,8.290)%
  --(4.845,8.290)--(4.854,8.290)--(4.863,8.290)--(4.872,8.291)--(4.881,8.291)--(4.890,8.291)%
  --(4.899,8.291)--(4.908,8.291)--(4.917,8.290)--(4.926,8.290)--(4.935,8.290)--(4.944,8.290)%
  --(4.953,8.291)--(4.962,8.291)--(4.971,8.291)--(4.980,8.291)--(4.989,8.291)--(4.998,8.292)%
  --(5.007,8.292)--(5.016,8.292)--(5.025,8.292)--(5.034,8.292)--(5.043,8.293)--(5.052,8.293)%
  --(5.061,8.293)--(5.070,8.293)--(5.079,8.294)--(5.088,8.294)--(5.096,8.294)--(5.105,8.294)%
  --(5.114,8.294)--(5.123,8.294)--(5.132,8.295)--(5.141,8.295)--(5.150,8.295)--(5.159,8.295)%
  --(5.168,8.295)--(5.177,8.296)--(5.186,8.296)--(5.195,8.296)--(5.204,8.296)--(5.213,8.296)%
  --(5.222,8.297)--(5.231,8.297)--(5.240,8.297)--(5.249,8.297)--(5.258,8.297)--(5.267,8.298)%
  --(5.276,8.298)--(5.285,8.298)--(5.294,8.298)--(5.303,8.298)--(5.312,8.298)--(5.321,8.299)%
  --(5.330,8.299)--(5.339,8.299)--(5.348,8.299)--(5.357,8.299)--(5.365,8.299)--(5.374,8.300)%
  --(5.383,8.300)--(5.392,8.300)--(5.401,8.300)--(5.410,8.300)--(5.419,8.301)--(5.428,8.301)%
  --(5.437,8.301)--(5.446,8.301)--(5.455,8.301)--(5.464,8.301)--(5.473,8.302)--(5.482,8.302)%
  --(5.491,8.302)--(5.500,8.302)--(5.509,8.302)--(5.518,8.302)--(5.527,8.303)--(5.536,8.303)%
  --(5.545,8.303)--(5.554,8.303)--(5.563,8.303)--(5.572,8.303)--(5.581,8.303)--(5.590,8.304)%
  --(5.599,8.304)--(5.608,8.304)--(5.617,8.304)--(5.626,8.304)--(5.634,8.304)--(5.643,8.305)%
  --(5.652,8.305)--(5.661,8.305)--(5.670,8.305)--(5.679,8.305)--(5.688,8.305)--(5.697,8.305)%
  --(5.706,8.306)--(5.715,8.306)--(5.724,8.306)--(5.733,8.306)--(5.742,8.306)--(5.751,8.306)%
  --(5.760,8.307)--(5.769,8.307)--(5.778,8.307)--(5.787,8.307)--(5.796,8.307)--(5.805,8.307)%
  --(5.814,8.307)--(5.823,8.308)--(5.832,8.308)--(5.841,8.308)--(5.850,8.308)--(5.859,8.308)%
  --(5.868,8.308)--(5.877,8.308)--(5.886,8.309)--(5.895,8.309)--(5.903,8.309)--(5.912,8.309)%
  --(5.921,8.309)--(5.930,8.309)--(5.939,8.309)--(5.948,8.309)--(5.957,8.310)--(5.966,8.310)%
  --(5.975,8.310)--(5.984,8.310)--(5.993,8.310)--(6.002,8.310)--(6.011,8.310)--(6.020,8.311)%
  --(6.029,8.311)--(6.038,8.311)--(6.047,8.311)--(6.056,8.311)--(6.065,8.311)--(6.074,8.311)%
  --(6.083,8.311)--(6.092,8.312)--(6.101,8.312)--(6.110,8.312)--(6.119,8.312)--(6.128,8.312)%
  --(6.137,8.312)--(6.146,8.312)--(6.155,8.312)--(6.164,8.313)--(6.172,8.313)--(6.181,8.313)%
  --(6.190,8.313)--(6.199,8.313)--(6.208,8.311)--(6.217,8.311)--(6.226,8.311)--(6.235,8.312)%
  --(6.244,8.312)--(6.253,8.312)--(6.262,8.312)--(6.271,8.312)--(6.280,8.312)--(6.289,8.312)%
  --(6.298,8.312)--(6.307,8.313)--(6.316,8.313)--(6.325,8.313)--(6.334,8.312)--(6.343,8.312)%
  --(6.352,8.312)--(6.361,8.313)--(6.370,8.313)--(6.379,8.313)--(6.388,8.313)--(6.397,8.313)%
  --(6.406,8.313)--(6.415,8.313)--(6.424,8.313)--(6.433,8.314)--(6.441,8.314)--(6.450,8.314)%
  --(6.459,8.314)--(6.468,8.314)--(6.477,8.314)--(6.486,8.314)--(6.495,8.314)--(6.504,8.314)%
  --(6.513,8.315)--(6.522,8.315)--(6.531,8.315)--(6.540,8.315)--(6.549,8.315)--(6.558,8.315)%
  --(6.567,8.313)--(6.576,8.313)--(6.585,8.313)--(6.594,8.313)--(6.603,8.313)--(6.612,8.313)%
  --(6.621,8.313)--(6.630,8.313)--(6.639,8.308)--(6.648,8.308)--(6.657,8.308)--(6.666,8.308)%
  --(6.675,8.308)--(6.684,8.309)--(6.693,8.309)--(6.702,8.309)--(6.710,8.309)--(6.719,8.309)%
  --(6.728,8.309)--(6.737,8.309)--(6.746,8.309)--(6.755,8.310)--(6.764,8.310)--(6.773,8.310)%
  --(6.782,8.310)--(6.791,8.310)--(6.800,8.310)--(6.809,8.310)--(6.818,8.310)--(6.827,8.310)%
  --(6.836,8.311)--(6.845,8.311)--(6.854,8.311)--(6.863,8.311)--(6.872,8.311)--(6.881,8.311)%
  --(6.890,8.311)--(6.899,8.311)--(6.908,8.311)--(6.917,8.312)--(6.926,8.312)--(6.935,8.312)%
  --(6.944,8.312)--(6.953,8.312)--(6.962,8.312)--(6.971,8.312)--(6.979,8.312)--(6.988,8.312)%
  --(6.997,8.312)--(7.006,8.312)--(7.015,8.312)--(7.024,8.313)--(7.033,8.313)--(7.042,8.313)%
  --(7.051,8.313)--(7.060,8.313)--(7.069,8.313)--(7.078,8.313)--(7.087,8.313)--(7.096,8.313)%
  --(7.105,8.314)--(7.114,8.314)--(7.123,8.314)--(7.132,8.314)--(7.141,8.314)--(7.150,8.314)%
  --(7.159,8.314)--(7.168,8.314)--(7.177,8.314)--(7.186,8.314)--(7.195,8.315)--(7.204,8.315)%
  --(7.213,8.315)--(7.222,8.315)--(7.231,8.315)--(7.240,8.315)--(7.248,8.315)--(7.257,8.315)%
  --(7.266,8.315)--(7.275,8.315)--(7.284,8.315)--(7.293,8.316)--(7.302,8.316)--(7.311,8.316)%
  --(7.320,8.316)--(7.329,8.316)--(7.338,8.316)--(7.347,8.316)--(7.356,8.316)--(7.365,8.316)%
  --(7.374,8.316)--(7.383,8.317)--(7.392,8.317)--(7.401,8.317)--(7.410,8.317)--(7.419,8.317)%
  --(7.428,8.317)--(7.437,8.317)--(7.446,8.317)--(7.455,8.317)--(7.464,8.317)--(7.473,8.317)%
  --(7.482,8.317)--(7.491,8.317)--(7.500,8.317)--(7.509,8.317)--(7.517,8.318)--(7.526,8.318)%
  --(7.535,8.318)--(7.544,8.318)--(7.553,8.318)--(7.562,8.318)--(7.571,8.318)--(7.580,8.318)%
  --(7.589,8.318)--(7.598,8.318)--(7.607,8.318)--(7.616,8.318)--(7.625,8.318)--(7.634,8.318)%
  --(7.643,8.319)--(7.652,8.319)--(7.661,8.319)--(7.670,8.317)--(7.679,8.317)--(7.688,8.317)%
  --(7.697,8.317)--(7.706,8.317)--(7.715,8.317)--(7.724,8.318)--(7.733,8.318)--(7.742,8.318)%
  --(7.751,8.318)--(7.760,8.318)--(7.769,8.318)--(7.778,8.318)--(7.786,8.318)--(7.795,8.318)%
  --(7.804,8.318)--(7.813,8.318)--(7.822,8.318)--(7.831,8.319)--(7.840,8.319)--(7.849,8.319)%
  --(7.858,8.319)--(7.867,8.319)--(7.876,8.319)--(7.885,8.319)--(7.894,8.319)--(7.903,8.319)%
  --(7.912,8.319)--(7.921,8.319)--(7.930,8.319)--(7.939,8.320)--(7.948,8.320)--(7.957,8.320)%
  --(7.966,8.320)--(7.975,8.320)--(7.984,8.320)--(7.993,8.320)--(8.002,8.320)--(8.011,8.320)%
  --(8.020,8.320)--(8.029,8.320)--(8.038,8.320)--(8.047,8.321)--(8.055,8.321)--(8.064,8.321)%
  --(8.073,8.321)--(8.082,8.321)--(8.091,8.321)--(8.100,8.321)--(8.109,8.321)--(8.118,8.321)%
  --(8.127,8.321)--(8.136,8.321)--(8.145,8.321)--(8.154,8.321)--(8.163,8.322)--(8.172,8.322)%
  --(8.181,8.322)--(8.190,8.322)--(8.199,8.322)--(8.208,8.322)--(8.217,8.322)--(8.226,8.322)%
  --(8.235,8.322)--(8.244,8.322)--(8.253,8.322)--(8.262,8.322)--(8.271,8.322)--(8.280,8.322)%
  --(8.289,8.323)--(8.298,8.323)--(8.307,8.323)--(8.316,8.323)--(8.324,8.323)--(8.333,8.323)%
  --(8.342,8.323)--(8.351,8.323)--(8.360,8.323)--(8.369,8.323)--(8.378,8.323)--(8.387,8.323)%
  --(8.396,8.323)--(8.405,8.324)--(8.414,8.324)--(8.423,8.324)--(8.432,8.321)--(8.441,8.321)%
  --(8.450,8.321)--(8.459,8.321)--(8.468,8.321)--(8.477,8.321)--(8.486,8.322)--(8.495,8.322)%
  --(8.504,8.322)--(8.513,8.322)--(8.522,8.322)--(8.531,8.322)--(8.540,8.322)--(8.549,8.322)%
  --(8.558,8.322)--(8.567,8.322)--(8.576,8.322)--(8.585,8.322)--(8.593,8.322)--(8.602,8.322)%
  --(8.611,8.323)--(8.620,8.323)--(8.629,8.323)--(8.638,8.323)--(8.647,8.323)--(8.656,8.323)%
  --(8.665,8.323)--(8.674,8.323)--(8.683,8.323)--(8.692,8.323)--(8.701,8.323)--(8.710,8.323)%
  --(8.719,8.323)--(8.728,8.323)--(8.737,8.323)--(8.746,8.324)--(8.755,8.324)--(8.764,8.324)%
  --(8.773,8.324)--(8.782,8.324)--(8.791,8.324)--(8.800,8.324)--(8.809,8.324)--(8.818,8.324)%
  --(8.827,8.324)--(8.836,8.324)--(8.845,8.324)--(8.854,8.324)--(8.862,8.324)--(8.871,8.324)%
  --(8.880,8.324)--(8.889,8.324)--(8.898,8.324)--(8.907,8.324)--(8.916,8.324)--(8.925,8.324)%
  --(8.934,8.324)--(8.943,8.324)--(8.952,8.324)--(8.961,8.324)--(8.970,8.324)--(8.979,8.325)%
  --(8.988,8.325)--(8.997,8.325)--(9.006,8.325)--(9.015,8.325)--(9.024,8.325)--(9.033,8.325)%
  --(9.042,8.325)--(9.051,8.325)--(9.060,8.325)--(9.069,8.325)--(9.078,8.325)--(9.087,8.325)%
  --(9.096,8.325)--(9.105,8.325)--(9.114,8.325)--(9.123,8.325)--(9.131,8.325)--(9.140,8.325)%
  --(9.149,8.325)--(9.158,8.325)--(9.167,8.325)--(9.176,8.325)--(9.185,8.326)--(9.194,8.326)%
  --(9.203,8.326)--(9.212,8.326)--(9.221,8.326)--(9.230,8.326)--(9.239,8.326)--(9.248,8.326)%
  --(9.257,8.326)--(9.266,8.326)--(9.275,8.326)--(9.284,8.326)--(9.293,8.326)--(9.302,8.326)%
  --(9.311,8.326)--(9.320,8.326)--(9.329,8.326)--(9.338,8.327)--(9.347,8.327)--(9.356,8.327)%
  --(9.365,8.327)--(9.374,8.327)--(9.383,8.327)--(9.392,8.327)--(9.400,8.327)--(9.409,8.327)%
  --(9.418,8.327)--(9.427,8.327)--(9.436,8.327)--(9.445,8.327)--(9.454,8.327)--(9.463,8.327)%
  --(9.472,8.327)--(9.481,8.327)--(9.490,8.328)--(9.499,8.328)--(9.508,8.328)--(9.517,8.328)%
  --(9.526,8.328)--(9.535,8.328)--(9.544,8.328)--(9.553,8.328)--(9.562,8.328)--(9.571,8.328)%
  --(9.580,8.328)--(9.589,8.328)--(9.598,8.328)--(9.607,8.328)--(9.616,8.328)--(9.625,8.328)%
  --(9.634,8.328)--(9.643,8.329)--(9.652,8.329)--(9.661,8.329)--(9.669,8.329)--(9.678,8.329)%
  --(9.687,8.329)--(9.696,8.329)--(9.705,8.329)--(9.714,8.329)--(9.723,8.329)--(9.732,8.329)%
  --(9.741,8.329)--(9.750,8.329)--(9.759,8.329)--(9.768,8.329)--(9.777,8.329)--(9.786,8.329)%
  --(9.795,8.329)--(9.804,8.329)--(9.813,8.330)--(9.822,8.330)--(9.831,8.330)--(9.840,8.330)%
  --(9.849,8.330)--(9.858,8.330)--(9.867,8.330)--(9.876,8.330)--(9.885,8.330)--(9.894,8.330)%
  --(9.903,8.330)--(9.912,8.330)--(9.921,8.330)--(9.930,8.330)--(9.938,8.330)--(9.947,8.330)%
  --(9.956,8.330)--(9.965,8.330)--(9.974,8.330)--(9.983,8.331)--(9.992,8.331)--(10.001,8.331)%
  --(10.010,8.331)--(10.019,8.331)--(10.028,8.331)--(10.037,8.331)--(10.046,8.331)--(10.055,8.331)%
  --(10.064,8.331)--(10.073,8.331)--(10.082,8.331)--(10.091,8.331)--(10.100,8.331)--(10.109,8.331)%
  --(10.118,8.331)--(10.127,8.331)--(10.136,8.331)--(10.145,8.331)--(10.154,8.332)--(10.163,8.332)%
  --(10.172,8.332)--(10.181,8.332)--(10.190,8.332)--(10.199,8.332)--(10.207,8.332)--(10.216,8.332)%
  --(10.225,8.332)--(10.234,8.332)--(10.243,8.332)--(10.252,8.332)--(10.261,8.332)--(10.270,8.332)%
  --(10.279,8.332)--(10.288,8.332)--(10.297,8.332)--(10.306,8.332)--(10.315,8.332)--(10.324,8.332)%
  --(10.333,8.332)--(10.342,8.333)--(10.351,8.333)--(10.360,8.333)--(10.369,8.333)--(10.378,8.333)%
  --(10.387,8.333)--(10.396,8.333)--(10.405,8.333)--(10.414,8.333)--(10.423,8.333)--(10.432,8.333)%
  --(10.441,8.333)--(10.450,8.333)--(10.459,8.333)--(10.468,8.333)--(10.476,8.333)--(10.485,8.333)%
  --(10.494,8.333)--(10.503,8.333)--(10.512,8.333)--(10.521,8.333)--(10.530,8.333)--(10.539,8.334)%
  --(10.548,8.334)--(10.557,8.334)--(10.566,8.334)--(10.575,8.334)--(10.584,8.334)--(10.593,8.334)%
  --(10.602,8.334)--(10.611,8.334)--(10.620,8.334)--(10.629,8.334)--(10.638,8.334)--(10.647,8.334)%
  --(10.656,8.334)--(10.665,8.334)--(10.674,8.334)--(10.683,8.334)--(10.692,8.334)--(10.701,8.334)%
  --(10.710,8.334)--(10.719,8.334)--(10.728,8.334)--(10.737,8.335)--(10.745,8.335)--(10.754,8.335)%
  --(10.763,8.335)--(10.772,8.335)--(10.781,8.335)--(10.790,8.335)--(10.799,8.335)--(10.808,8.334)%
  --(10.817,8.334)--(10.826,8.334)--(10.835,8.334)--(10.844,8.334)--(10.853,8.334)--(10.862,8.334)%
  --(10.871,8.334)--(10.880,8.334)--(10.889,8.334)--(10.898,8.334)--(10.907,8.334)--(10.916,8.334)%
  --(10.925,8.334)--(10.934,8.334)--(10.943,8.335)--(10.952,8.335)--(10.961,8.335)--(10.970,8.335)%
  --(10.979,8.335)--(10.988,8.335)--(10.997,8.335)--(11.006,8.335)--(11.014,8.335)--(11.023,8.335)%
  --(11.032,8.335)--(11.041,8.335)--(11.050,8.335)--(11.059,8.335)--(11.068,8.335)--(11.077,8.335)%
  --(11.086,8.335)--(11.095,8.335)--(11.104,8.335)--(11.113,8.335)--(11.122,8.335)--(11.131,8.335)%
  --(11.140,8.335)--(11.149,8.335)--(11.158,8.336)--(11.167,8.336)--(11.176,8.336)--(11.185,8.336)%
  --(11.194,8.336)--(11.203,8.336)--(11.212,8.336)--(11.221,8.336)--(11.230,8.336)--(11.239,8.336)%
  --(11.248,8.336)--(11.257,8.336)--(11.266,8.336)--(11.275,8.336)--(11.283,8.336)--(11.292,8.336)%
  --(11.301,8.336)--(11.310,8.336)--(11.319,8.336)--(11.328,8.336)--(11.337,8.336)--(11.346,8.336)%
  --(11.355,8.336)--(11.364,8.336)--(11.373,8.336)--(11.382,8.337)--(11.391,8.337)--(11.400,8.337)%
  --(11.409,8.337)--(11.418,8.337)--(11.427,8.337)--(11.436,8.337)--(11.445,8.337)--(11.454,8.337)%
  --(11.463,8.337)--(11.472,8.337)--(11.481,8.337)--(11.490,8.337)--(11.499,8.337)--(11.508,8.337)%
  --(11.517,8.337)--(11.526,8.337)--(11.535,8.337)--(11.544,8.337)--(11.552,8.337)--(11.561,8.337)%
  --(11.570,8.337)--(11.579,8.337)--(11.588,8.337)--(11.597,8.337)--(11.606,8.337)--(11.615,8.338)%
  --(11.624,8.338)--(11.633,8.338)--(11.642,8.338)--(11.651,8.338)--(11.660,8.338)--(11.669,8.338)%
  --(11.678,8.338)--(11.687,8.338)--(11.696,8.338)--(11.705,8.338)--(11.714,8.338)--(11.723,8.338)%
  --(11.732,8.338)--(11.741,8.338)--(11.750,8.338)--(11.759,8.338)--(11.768,8.338)--(11.777,8.338)%
  --(11.786,8.338)--(11.795,8.338)--(11.804,8.338)--(11.813,8.338)--(11.821,8.338)--(11.830,8.338)%
  --(11.839,8.338)--(11.848,8.338)--(11.857,8.338)--(11.866,8.338)--(11.875,8.338)--(11.884,8.338)%
  --(11.893,8.339)--(11.902,8.339)--(11.911,8.339)--(11.920,8.339)--(11.929,8.339)--(11.938,8.339)%
  --(11.947,8.339);
\end{tikzpicture}
}
\hspace{2cm}
\subfloat{  
\begin{tikzpicture}[scale=0.5]
\draw (1.196,0.368)--(1.376,0.368);
\node at (0.55,0.368) {\small 0.4};
\draw (1.196,3.039)--(1.376,3.039);
\node at (0.55,3.039) {\small 0.6};
\draw (1.196,5.710)--(1.376,5.710);
\node at (0.55,5.710) {\small 0.8};
\node at (0.55,8.381) {\small 1};
\node at (11.3,-0.2) {\small{time}};
\draw[->,>=stealth', color=blue](1.196,0.368)-- (1.196,8.8);
\draw[->,>=stealth', color=blue](1.196,0.368)--(12.5,0.368);
\node at (10.663,7.5) {\small{$v_n^1(E)$}};
\draw (1.196,0.435)--(1.205,3.084)--(1.214,4.408)--(1.223,5.203)--(1.232,5.732)%
  --(1.241,6.111)--(1.250,6.394)--(1.259,6.615)--(1.268,6.792)--(1.277,6.936)--(1.286,7.057)%
  --(1.295,7.159)--(1.304,7.246)--(1.313,7.322)--(1.322,7.388)--(1.330,7.446)--(1.339,7.498)%
  --(1.348,7.545)--(1.357,7.586)--(1.366,7.624)--(1.375,7.659)--(1.384,7.690)--(1.393,7.719)%
  --(1.402,7.745)--(1.411,7.770)--(1.420,7.792)--(1.429,7.813)--(1.438,7.833)--(1.447,7.851)%
  --(1.456,7.868)--(1.465,7.884)--(1.474,7.899)--(1.483,7.914)--(1.492,7.927)--(1.501,7.940)%
  --(1.510,7.951)--(1.519,7.963)--(1.528,7.974)--(1.537,7.984)--(1.546,7.993)--(1.555,8.003)%
  --(1.564,8.011)--(1.573,8.020)--(1.582,8.028)--(1.591,8.036)--(1.599,8.043)--(1.608,8.050)%
  --(1.617,8.057)--(1.626,8.063)--(1.635,8.069)--(1.644,8.075)--(1.653,8.081)--(1.662,8.087)%
  --(1.671,8.092)--(1.680,8.097)--(1.689,8.102)--(1.698,8.107)--(1.707,8.112)--(1.716,8.116)%
  --(1.725,8.120)--(1.734,8.125)--(1.743,8.129)--(1.752,8.133)--(1.761,8.137)--(1.770,8.140)%
  --(1.779,8.144)--(1.788,8.147)--(1.797,8.151)--(1.806,8.154)--(1.815,8.157)--(1.824,8.160)%
  --(1.833,8.163)--(1.842,8.166)--(1.851,8.169)--(1.860,8.172)--(1.868,8.175)--(1.877,8.177)%
  --(1.886,8.180)--(1.895,8.182)--(1.904,8.185)--(1.913,8.187)--(1.922,8.190)--(1.931,8.192)%
  --(1.940,8.194)--(1.949,8.196)--(1.958,8.198)--(1.967,8.200)--(1.976,8.202)--(1.985,8.204)%
  --(1.994,8.206)--(2.003,8.208)--(2.012,8.210)--(2.021,8.212)--(2.030,8.214)--(2.039,8.215)%
  --(2.048,8.217)--(2.057,8.219)--(2.066,8.220)--(2.075,8.222)--(2.084,8.224)--(2.093,8.225)%
  --(2.102,8.227)--(2.111,8.228)--(2.120,8.230)--(2.129,8.231)--(2.137,8.232)--(2.146,8.234)%
  --(2.155,8.235)--(2.164,8.237)--(2.173,8.238)--(2.182,8.239)--(2.191,8.240)--(2.200,8.242)%
  --(2.209,8.243)--(2.218,8.244)--(2.227,8.245)--(2.236,8.246)--(2.245,8.247)--(2.254,8.249)%
  --(2.263,8.250)--(2.272,8.251)--(2.281,8.252)--(2.290,8.253)--(2.299,8.254)--(2.308,8.255)%
  --(2.317,8.256)--(2.326,8.257)--(2.335,8.258)--(2.344,8.259)--(2.353,8.260)--(2.362,8.261)%
  --(2.371,8.262)--(2.380,8.262)--(2.389,8.263)--(2.398,8.264)--(2.406,8.265)--(2.415,8.266)%
  --(2.424,8.267)--(2.433,8.267)--(2.442,8.268)--(2.451,8.269)--(2.460,8.270)--(2.469,8.271)%
  --(2.478,8.271)--(2.487,8.272)--(2.496,8.273)--(2.505,8.274)--(2.514,8.274)--(2.523,8.275)%
  --(2.532,8.276)--(2.541,8.276)--(2.550,8.277)--(2.559,8.278)--(2.568,8.278)--(2.577,8.279)%
  --(2.586,8.280)--(2.595,8.280)--(2.604,8.281)--(2.613,8.282)--(2.622,8.282)--(2.631,8.283)%
  --(2.640,8.284)--(2.649,8.284)--(2.658,8.285)--(2.667,8.285)--(2.675,8.286)--(2.684,8.286)%
  --(2.693,8.287)--(2.702,8.288)--(2.711,8.288)--(2.720,8.289)--(2.729,8.289)--(2.738,8.290)%
  --(2.747,8.290)--(2.756,8.291)--(2.765,8.291)--(2.774,8.292)--(2.783,8.292)--(2.792,8.293)%
  --(2.801,8.293)--(2.810,8.294)--(2.819,8.294)--(2.828,8.295)--(2.837,8.295)--(2.846,8.296)%
  --(2.855,8.296)--(2.864,8.296)--(2.873,8.297)--(2.882,8.297)--(2.891,8.298)--(2.900,8.298)%
  --(2.909,8.299)--(2.918,8.299)--(2.927,8.300)--(2.936,8.300)--(2.944,8.300)--(2.953,8.301)%
  --(2.962,8.301)--(2.971,8.302)--(2.980,8.302)--(2.989,8.302)--(2.998,8.303)--(3.007,8.303)%
  --(3.016,8.303)--(3.025,8.304)--(3.034,8.304)--(3.043,8.305)--(3.052,8.305)--(3.061,8.305)%
  --(3.070,8.306)--(3.079,8.306)--(3.088,8.306)--(3.097,8.307)--(3.106,8.307)--(3.115,8.307)%
  --(3.124,8.308)--(3.133,8.308)--(3.142,8.308)--(3.151,8.309)--(3.160,8.309)--(3.169,8.309)%
  --(3.178,8.310)--(3.187,8.310)--(3.196,8.310)--(3.205,8.311)--(3.213,8.311)--(3.222,8.311)%
  --(3.231,8.312)--(3.240,8.312)--(3.249,8.312)--(3.258,8.312)--(3.267,8.313)--(3.276,8.313)%
  --(3.285,8.313)--(3.294,8.314)--(3.303,8.314)--(3.312,8.314)--(3.321,8.315)--(3.330,8.315)%
  --(3.339,8.315)--(3.348,8.315)--(3.357,8.316)--(3.366,8.316)--(3.375,8.316)--(3.384,8.316)%
  --(3.393,8.317)--(3.402,8.317)--(3.411,8.317)--(3.420,8.317)--(3.429,8.318)--(3.438,8.318)%
  --(3.447,8.318)--(3.456,8.318)--(3.465,8.319)--(3.474,8.319)--(3.482,8.319)--(3.491,8.319)%
  --(3.500,8.320)--(3.509,8.320)--(3.518,8.320)--(3.527,8.320)--(3.536,8.321)--(3.545,8.321)%
  --(3.554,8.321)--(3.563,8.321)--(3.572,8.321)--(3.581,8.322)--(3.590,8.322)--(3.599,8.322)%
  --(3.608,8.322)--(3.617,8.323)--(3.626,8.323)--(3.635,8.323)--(3.644,8.323)--(3.653,8.323)%
  --(3.662,8.324)--(3.671,8.324)--(3.680,8.324)--(3.689,8.324)--(3.698,8.324)--(3.707,8.325)%
  --(3.716,8.325)--(3.725,8.325)--(3.734,8.325)--(3.743,8.325)--(3.751,8.326)--(3.760,8.326)%
  --(3.769,8.326)--(3.778,8.326)--(3.787,8.326)--(3.796,8.327)--(3.805,8.327)--(3.814,8.327)%
  --(3.823,8.327)--(3.832,8.327)--(3.841,8.327)--(3.850,8.328)--(3.859,8.328)--(3.868,8.328)%
  --(3.877,8.328)--(3.886,8.328)--(3.895,8.329)--(3.904,8.329)--(3.913,8.329)--(3.922,8.329)%
  --(3.931,8.329)--(3.940,8.329)--(3.949,8.330)--(3.958,8.330)--(3.967,8.330)--(3.976,8.330)%
  --(3.985,8.330)--(3.994,8.330)--(4.003,8.331)--(4.012,8.331)--(4.020,8.331)--(4.029,8.331)%
  --(4.038,8.331)--(4.047,8.331)--(4.056,8.331)--(4.065,8.332)--(4.074,8.332)--(4.083,8.332)%
  --(4.092,8.332)--(4.101,8.332)--(4.110,8.332)--(4.119,8.333)--(4.128,8.333)--(4.137,8.333)%
  --(4.146,8.333)--(4.155,8.333)--(4.164,8.333)--(4.173,8.333)--(4.182,8.334)--(4.191,8.334)%
  --(4.200,8.334)--(4.209,8.334)--(4.218,8.334)--(4.227,8.334)--(4.236,8.334)--(4.245,8.335)%
  --(4.254,8.335)--(4.263,8.335)--(4.272,8.335)--(4.281,8.335)--(4.289,8.335)--(4.298,8.335)%
  --(4.307,8.335)--(4.316,8.336)--(4.325,8.336)--(4.334,8.336)--(4.343,8.336)--(4.352,8.336)%
  --(4.361,8.336)--(4.370,8.336)--(4.379,8.336)--(4.388,8.337)--(4.397,8.337)--(4.406,8.337)%
  --(4.415,8.337)--(4.424,8.337)--(4.433,8.337)--(4.442,8.337)--(4.451,8.337)--(4.460,8.338)%
  --(4.469,8.338)--(4.478,8.338)--(4.487,8.338)--(4.496,8.338)--(4.505,8.338)--(4.514,8.338)%
  --(4.523,8.338)--(4.532,8.339)--(4.541,8.339)--(4.550,8.339)--(4.558,8.339)--(4.567,8.339)%
  --(4.576,8.339)--(4.585,8.339)--(4.594,8.339)--(4.603,8.339)--(4.612,8.340)--(4.621,8.340)%
  --(4.630,8.340)--(4.639,8.340)--(4.648,8.340)--(4.657,8.340)--(4.666,8.340)--(4.675,8.340)%
  --(4.684,8.340)--(4.693,8.340)--(4.702,8.341)--(4.711,8.341)--(4.720,8.341)--(4.729,8.341)%
  --(4.738,8.341)--(4.747,8.341)--(4.756,8.341)--(4.765,8.341)--(4.774,8.341)--(4.783,8.341)%
  --(4.792,8.342)--(4.801,8.342)--(4.810,8.342)--(4.819,8.342)--(4.827,8.342)--(4.836,8.342)%
  --(4.845,8.342)--(4.854,8.342)--(4.863,8.342)--(4.872,8.342)--(4.881,8.343)--(4.890,8.343)%
  --(4.899,8.343)--(4.908,8.343)--(4.917,8.343)--(4.926,8.343)--(4.935,8.343)--(4.944,8.343)%
  --(4.953,8.343)--(4.962,8.343)--(4.971,8.343)--(4.980,8.344)--(4.989,8.344)--(4.998,8.344)%
  --(5.007,8.344)--(5.016,8.344)--(5.025,8.344)--(5.034,8.344)--(5.043,8.344)--(5.052,8.344)%
  --(5.061,8.344)--(5.070,8.344)--(5.079,8.344)--(5.088,8.345)--(5.096,8.345)--(5.105,8.345)%
  --(5.114,8.345)--(5.123,8.345)--(5.132,8.345)--(5.141,8.345)--(5.150,8.345)--(5.159,8.345)%
  --(5.168,8.345)--(5.177,8.345)--(5.186,8.345)--(5.195,8.346)--(5.204,8.346)--(5.213,8.346)%
  --(5.222,8.346)--(5.231,8.346)--(5.240,8.346)--(5.249,8.346)--(5.258,8.346)--(5.267,8.346)%
  --(5.276,8.346)--(5.285,8.346)--(5.294,8.346)--(5.303,8.346)--(5.312,8.347)--(5.321,8.347)%
  --(5.330,8.347)--(5.339,8.347)--(5.348,8.347)--(5.357,8.347)--(5.365,8.347)--(5.374,8.347)%
  --(5.383,8.347)--(5.392,8.347)--(5.401,8.347)--(5.410,8.347)--(5.419,8.347)--(5.428,8.347)%
  --(5.437,8.348)--(5.446,8.348)--(5.455,8.348)--(5.464,8.348)--(5.473,8.348)--(5.482,8.348)%
  --(5.491,8.348)--(5.500,8.348)--(5.509,8.348)--(5.518,8.348)--(5.527,8.348)--(5.536,8.348)%
  --(5.545,8.348)--(5.554,8.348)--(5.563,8.349)--(5.572,8.349)--(5.581,8.349)--(5.590,8.349)%
  --(5.599,8.349)--(5.608,8.349)--(5.617,8.349)--(5.626,8.349)--(5.634,8.349)--(5.643,8.349)%
  --(5.652,8.349)--(5.661,8.349)--(5.670,8.349)--(5.679,8.349)--(5.688,8.349)--(5.697,8.349)%
  --(5.706,8.350)--(5.715,8.350)--(5.724,8.350)--(5.733,8.350)--(5.742,8.350)--(5.751,8.350)%
  --(5.760,8.350)--(5.769,8.350)--(5.778,8.350)--(5.787,8.350)--(5.796,8.350)--(5.805,8.350)%
  --(5.814,8.350)--(5.823,8.350)--(5.832,8.350)--(5.841,8.350)--(5.850,8.350)--(5.859,8.351)%
  --(5.868,8.351)--(5.877,8.351)--(5.886,8.351)--(5.895,8.351)--(5.903,8.351)--(5.912,8.351)%
  --(5.921,8.351)--(5.930,8.351)--(5.939,8.351)--(5.948,8.351)--(5.957,8.351)--(5.966,8.351)%
  --(5.975,8.351)--(5.984,8.351)--(5.993,8.351)--(6.002,8.351)--(6.011,8.352)--(6.020,8.352)%
  --(6.029,8.352)--(6.038,8.352)--(6.047,8.352)--(6.056,8.352)--(6.065,8.352)--(6.074,8.352)%
  --(6.083,8.352)--(6.092,8.352)--(6.101,8.352)--(6.110,8.352)--(6.119,8.352)--(6.128,8.352)%
  --(6.137,8.352)--(6.146,8.352)--(6.155,8.352)--(6.164,8.352)--(6.172,8.352)--(6.181,8.353)%
  --(6.190,8.353)--(6.199,8.353)--(6.208,8.353)--(6.217,8.353)--(6.226,8.353)--(6.235,8.353)%
  --(6.244,8.353)--(6.253,8.353)--(6.262,8.353)--(6.271,8.353)--(6.280,8.353)--(6.289,8.353)%
  --(6.298,8.353)--(6.307,8.353)--(6.316,8.353)--(6.325,8.353)--(6.334,8.353)--(6.343,8.353)%
  --(6.352,8.353)--(6.361,8.354)--(6.370,8.354)--(6.379,8.354)--(6.388,8.354)--(6.397,8.354)%
  --(6.406,8.354)--(6.415,8.354)--(6.424,8.354)--(6.433,8.354)--(6.441,8.354)--(6.450,8.354)%
  --(6.459,8.354)--(6.468,8.354)--(6.477,8.354)--(6.486,8.354)--(6.495,8.354)--(6.504,8.354)%
  --(6.513,8.354)--(6.522,8.354)--(6.531,8.354)--(6.540,8.354)--(6.549,8.354)--(6.558,8.355)%
  --(6.567,8.355)--(6.576,8.355)--(6.585,8.355)--(6.594,8.355)--(6.603,8.355)--(6.612,8.355)%
  --(6.621,8.355)--(6.630,8.355)--(6.639,8.355)--(6.648,8.355)--(6.657,8.355)--(6.666,8.355)%
  --(6.675,8.355)--(6.684,8.355)--(6.693,8.355)--(6.702,8.355)--(6.710,8.355)--(6.719,8.355)%
  --(6.728,8.355)--(6.737,8.355)--(6.746,8.355)--(6.755,8.355)--(6.764,8.355)--(6.773,8.356)%
  --(6.782,8.356)--(6.791,8.356)--(6.800,8.356)--(6.809,8.356)--(6.818,8.356)--(6.827,8.356)%
  --(6.836,8.356)--(6.845,8.356)--(6.854,8.356)--(6.863,8.356)--(6.872,8.356)--(6.881,8.356)%
  --(6.890,8.356)--(6.899,8.356)--(6.908,8.356)--(6.917,8.356)--(6.926,8.356)--(6.935,8.356)%
  --(6.944,8.356)--(6.953,8.356)--(6.962,8.356)--(6.971,8.356)--(6.979,8.356)--(6.988,8.356)%
  --(6.997,8.357)--(7.006,8.357)--(7.015,8.357)--(7.024,8.357)--(7.033,8.357)--(7.042,8.357)%
  --(7.051,8.357)--(7.060,8.357)--(7.069,8.357)--(7.078,8.357)--(7.087,8.357)--(7.096,8.357)%
  --(7.105,8.357)--(7.114,8.357)--(7.123,8.357)--(7.132,8.357)--(7.141,8.357)--(7.150,8.357)%
  --(7.159,8.357)--(7.168,8.357)--(7.177,8.357)--(7.186,8.357)--(7.195,8.357)--(7.204,8.357)%
  --(7.213,8.357)--(7.222,8.357)--(7.231,8.357)--(7.240,8.357)--(7.248,8.358)--(7.257,8.358)%
  --(7.266,8.358)--(7.275,8.358)--(7.284,8.358)--(7.293,8.358)--(7.302,8.358)--(7.311,8.358)%
  --(7.320,8.358)--(7.329,8.358)--(7.338,8.358)--(7.347,8.358)--(7.356,8.358)--(7.365,8.358)%
  --(7.374,8.358)--(7.383,8.358)--(7.392,8.358)--(7.401,8.358)--(7.410,8.358)--(7.419,8.358)%
  --(7.428,8.358)--(7.437,8.358)--(7.446,8.358)--(7.455,8.358)--(7.464,8.358)--(7.473,8.358)%
  --(7.482,8.358)--(7.491,8.358)--(7.500,8.358)--(7.509,8.358)--(7.517,8.359)--(7.526,8.359)%
  --(7.535,8.359)--(7.544,8.359)--(7.553,8.359)--(7.562,8.359)--(7.571,8.359)--(7.580,8.359)%
  --(7.589,8.359)--(7.598,8.359)--(7.607,8.359)--(7.616,8.359)--(7.625,8.359)--(7.634,8.359)%
  --(7.643,8.359)--(7.652,8.359)--(7.661,8.359)--(7.670,8.359)--(7.679,8.359)--(7.688,8.359)%
  --(7.697,8.359)--(7.706,8.359)--(7.715,8.359)--(7.724,8.359)--(7.733,8.359)--(7.742,8.359)%
  --(7.751,8.359)--(7.760,8.359)--(7.769,8.359)--(7.778,8.359)--(7.786,8.359)--(7.795,8.359)%
  --(7.804,8.359)--(7.813,8.360)--(7.822,8.360)--(7.831,8.360)--(7.840,8.360)--(7.849,8.360)%
  --(7.858,8.360)--(7.867,8.360)--(7.876,8.360)--(7.885,8.360)--(7.894,8.360)--(7.903,8.360)%
  --(7.912,8.360)--(7.921,8.360)--(7.930,8.360)--(7.939,8.360)--(7.948,8.360)--(7.957,8.360)%
  --(7.966,8.360)--(7.975,8.360)--(7.984,8.360)--(7.993,8.360)--(8.002,8.360)--(8.011,8.360)%
  --(8.020,8.360)--(8.029,8.360)--(8.038,8.360)--(8.047,8.360)--(8.055,8.360)--(8.064,8.360)%
  --(8.073,8.360)--(8.082,8.360)--(8.091,8.360)--(8.100,8.360)--(8.109,8.360)--(8.118,8.360)%
  --(8.127,8.360)--(8.136,8.361)--(8.145,8.361)--(8.154,8.361)--(8.163,8.361)--(8.172,8.361)%
  --(8.181,8.361)--(8.190,8.361)--(8.199,8.361)--(8.208,8.361)--(8.217,8.361)--(8.226,8.361)%
  --(8.235,8.361)--(8.244,8.361)--(8.253,8.361)--(8.262,8.361)--(8.271,8.361)--(8.280,8.361)%
  --(8.289,8.361)--(8.298,8.361)--(8.307,8.361)--(8.316,8.361)--(8.324,8.361)--(8.333,8.361)%
  --(8.342,8.361)--(8.351,8.361)--(8.360,8.361)--(8.369,8.361)--(8.378,8.361)--(8.387,8.361)%
  --(8.396,8.361)--(8.405,8.361)--(8.414,8.361)--(8.423,8.361)--(8.432,8.361)--(8.441,8.361)%
  --(8.450,8.361)--(8.459,8.361)--(8.468,8.361)--(8.477,8.361)--(8.486,8.362)--(8.495,8.362)%
  --(8.504,8.362)--(8.513,8.362)--(8.522,8.362)--(8.531,8.362)--(8.540,8.362)--(8.549,8.362)%
  --(8.558,8.362)--(8.567,8.362)--(8.576,8.362)--(8.585,8.362)--(8.593,8.362)--(8.602,8.362)%
  --(8.611,8.362)--(8.620,8.362)--(8.629,8.362)--(8.638,8.362)--(8.647,8.362)--(8.656,8.362)%
  --(8.665,8.362)--(8.674,8.362)--(8.683,8.362)--(8.692,8.362)--(8.701,8.362)--(8.710,8.362)%
  --(8.719,8.362)--(8.728,8.362)--(8.737,8.362)--(8.746,8.362)--(8.755,8.362)--(8.764,8.362)%
  --(8.773,8.362)--(8.782,8.362)--(8.791,8.362)--(8.800,8.362)--(8.809,8.362)--(8.818,8.362)%
  --(8.827,8.362)--(8.836,8.362)--(8.845,8.362)--(8.854,8.362)--(8.862,8.362)--(8.871,8.362)%
  --(8.880,8.362)--(8.889,8.363)--(8.898,8.363)--(8.907,8.363)--(8.916,8.363)--(8.925,8.363)%
  --(8.934,8.363)--(8.943,8.363)--(8.952,8.363)--(8.961,8.363)--(8.970,8.363)--(8.979,8.363)%
  --(8.988,8.363)--(8.997,8.363)--(9.006,8.363)--(9.015,8.363)--(9.024,8.363)--(9.033,8.363)%
  --(9.042,8.363)--(9.051,8.363)--(9.060,8.363)--(9.069,8.363)--(9.078,8.363)--(9.087,8.363)%
  --(9.096,8.363)--(9.105,8.363)--(9.114,8.363)--(9.123,8.363)--(9.131,8.363)--(9.140,8.363)%
  --(9.149,8.363)--(9.158,8.363)--(9.167,8.363)--(9.176,8.363)--(9.185,8.363)--(9.194,8.363)%
  --(9.203,8.363)--(9.212,8.363)--(9.221,8.363)--(9.230,8.363)--(9.239,8.363)--(9.248,8.363)%
  --(9.257,8.363)--(9.266,8.363)--(9.275,8.363)--(9.284,8.363)--(9.293,8.363)--(9.302,8.363)%
  --(9.311,8.363)--(9.320,8.363)--(9.329,8.364)--(9.338,8.364)--(9.347,8.364)--(9.356,8.364)%
  --(9.365,8.364)--(9.374,8.364)--(9.383,8.364)--(9.392,8.364)--(9.400,8.364)--(9.409,8.364)%
  --(9.418,8.364)--(9.427,8.364)--(9.436,8.364)--(9.445,8.364)--(9.454,8.364)--(9.463,8.364)%
  --(9.472,8.364)--(9.481,8.364)--(9.490,8.364)--(9.499,8.364)--(9.508,8.364)--(9.517,8.364)%
  --(9.526,8.364)--(9.535,8.364)--(9.544,8.364)--(9.553,8.364)--(9.562,8.364)--(9.571,8.364)%
  --(9.580,8.364)--(9.589,8.364)--(9.598,8.364)--(9.607,8.364)--(9.616,8.364)--(9.625,8.364)%
  --(9.634,8.364)--(9.643,8.364)--(9.652,8.364)--(9.661,8.364)--(9.669,8.364)--(9.678,8.364)%
  --(9.687,8.364)--(9.696,8.364)--(9.705,8.364)--(9.714,8.364)--(9.723,8.364)--(9.732,8.364)%
  --(9.741,8.364)--(9.750,8.364)--(9.759,8.364)--(9.768,8.364)--(9.777,8.364)--(9.786,8.364)%
  --(9.795,8.364)--(9.804,8.364)--(9.813,8.364)--(9.822,8.365)--(9.831,8.365)--(9.840,8.365)%
  --(9.849,8.365)--(9.858,8.365)--(9.867,8.365)--(9.876,8.365)--(9.885,8.365)--(9.894,8.365)%
  --(9.903,8.365)--(9.912,8.365)--(9.921,8.365)--(9.930,8.365)--(9.938,8.365)--(9.947,8.365)%
  --(9.956,8.365)--(9.965,8.365)--(9.974,8.365)--(9.983,8.365)--(9.992,8.365)--(10.001,8.365)%
  --(10.010,8.365)--(10.019,8.365)--(10.028,8.365)--(10.037,8.365)--(10.046,8.365)--(10.055,8.365)%
  --(10.064,8.365)--(10.073,8.365)--(10.082,8.365)--(10.091,8.365)--(10.100,8.365)--(10.109,8.365)%
  --(10.118,8.365)--(10.127,8.365)--(10.136,8.365)--(10.145,8.365)--(10.154,8.365)--(10.163,8.365)%
  --(10.172,8.365)--(10.181,8.365)--(10.190,8.365)--(10.199,8.365)--(10.207,8.365)--(10.216,8.365)%
  --(10.225,8.365)--(10.234,8.365)--(10.243,8.365)--(10.252,8.365)--(10.261,8.365)--(10.270,8.365)%
  --(10.279,8.365)--(10.288,8.365)--(10.297,8.365)--(10.306,8.365)--(10.315,8.365)--(10.324,8.365)%
  --(10.333,8.365)--(10.342,8.365)--(10.351,8.365)--(10.360,8.365)--(10.369,8.365)--(10.378,8.366)%
  --(10.387,8.366)--(10.396,8.366)--(10.405,8.366)--(10.414,8.366)--(10.423,8.366)--(10.432,8.366)%
  --(10.441,8.366)--(10.450,8.366)--(10.459,8.366)--(10.468,8.366)--(10.476,8.366)--(10.485,8.366)%
  --(10.494,8.366)--(10.503,8.366)--(10.512,8.366)--(10.521,8.366)--(10.530,8.366)--(10.539,8.366)%
  --(10.548,8.366)--(10.557,8.366)--(10.566,8.366)--(10.575,8.366)--(10.584,8.366)--(10.593,8.366)%
  --(10.602,8.366)--(10.611,8.366)--(10.620,8.366)--(10.629,8.366)--(10.638,8.366)--(10.647,8.366)%
  --(10.656,8.366)--(10.665,8.366)--(10.674,8.366)--(10.683,8.366)--(10.692,8.366)--(10.701,8.366)%
  --(10.710,8.366)--(10.719,8.366)--(10.728,8.366)--(10.737,8.366)--(10.745,8.366)--(10.754,8.366)%
  --(10.763,8.366)--(10.772,8.366)--(10.781,8.366)--(10.790,8.366)--(10.799,8.366)--(10.808,8.366)%
  --(10.817,8.366)--(10.826,8.366)--(10.835,8.366)--(10.844,8.366)--(10.853,8.366)--(10.862,8.366)%
  --(10.871,8.366)--(10.880,8.366)--(10.889,8.366)--(10.898,8.366)--(10.907,8.366)--(10.916,8.366)%
  --(10.925,8.366)--(10.934,8.366)--(10.943,8.366)--(10.952,8.366)--(10.961,8.366)--(10.970,8.366)%
  --(10.979,8.366)--(10.988,8.366)--(10.997,8.366)--(11.006,8.366)--(11.014,8.367)--(11.023,8.367)%
  --(11.032,8.367)--(11.041,8.367)--(11.050,8.367)--(11.059,8.367)--(11.068,8.367)--(11.077,8.367)%
  --(11.086,8.367)--(11.095,8.367)--(11.104,8.367)--(11.113,8.367)--(11.122,8.367)--(11.131,8.367)%
  --(11.140,8.367)--(11.149,8.367)--(11.158,8.367)--(11.167,8.367)--(11.176,8.367)--(11.185,8.367)%
  --(11.194,8.367)--(11.203,8.367)--(11.212,8.367)--(11.221,8.367)--(11.230,8.367)--(11.239,8.367)%
  --(11.248,8.367)--(11.257,8.367)--(11.266,8.367)--(11.275,8.367)--(11.283,8.367)--(11.292,8.367)%
  --(11.301,8.367)--(11.310,8.367)--(11.319,8.367)--(11.328,8.367)--(11.337,8.367)--(11.346,8.367)%
  --(11.355,8.367)--(11.364,8.367)--(11.373,8.367)--(11.382,8.367)--(11.391,8.367)--(11.400,8.367)%
  --(11.409,8.367)--(11.418,8.367)--(11.427,8.367)--(11.436,8.367)--(11.445,8.367)--(11.454,8.367)%
  --(11.463,8.367)--(11.472,8.367)--(11.481,8.367)--(11.490,8.367)--(11.499,8.367)--(11.508,8.367)%
  --(11.517,8.367)--(11.526,8.367)--(11.535,8.367)--(11.544,8.367)--(11.552,8.367)--(11.561,8.367)%
  --(11.570,8.367)--(11.579,8.367)--(11.588,8.367)--(11.597,8.367)--(11.606,8.367)--(11.615,8.367)%
  --(11.624,8.367)--(11.633,8.367)--(11.642,8.367)--(11.651,8.367)--(11.660,8.367)--(11.669,8.367)%
  --(11.678,8.367)--(11.687,8.367)--(11.696,8.367)--(11.705,8.367)--(11.714,8.367)--(11.723,8.367)%
  --(11.732,8.367)--(11.741,8.368)--(11.750,8.368)--(11.759,8.368)--(11.768,8.368)--(11.777,8.368)%
  --(11.786,8.368)--(11.795,8.368)--(11.804,8.368)--(11.813,8.368)--(11.821,8.368)--(11.830,8.368)%
  --(11.839,8.368)--(11.848,8.368)--(11.857,8.368)--(11.866,8.368)--(11.875,8.368)--(11.884,8.368)%
  --(11.893,8.368)--(11.902,8.368)--(11.911,8.368)--(11.920,8.368)--(11.929,8.368)--(11.938,8.368)%
  --(11.947,8.368);
\end{tikzpicture}  
 } 
\end{center}  
\vspace{-0.7cm}
\setcounter{subfigure}{0}
\begin{center}
   \subfloat[On top, $v_n^1(B) \to 1$. On the bottom, $\overline g_n^1 \to 1$.]{
\begin{tikzpicture}[scale=0.5]
\draw (1.196,0.368)--(1.376,0.368);
\node at (0.5,0.368) {\small{0}};
\draw (1.196,2.371)--(1.376,2.371);
\node at (0.5,2.371) {\small{ 0.5}};
\draw (1.196,4.375)--(1.376,4.375);
\node at (0.5,4.375) {\small{ 1}};
\draw (1.196,6.378)--(1.376,6.378);
\node at (0.5,6.378) {\small{ 1.5}};
\node at (0.5,8.381) {\small{2}};
\draw[->,>=stealth', color=blue](1.196,0.368)-- (1.196,8.381);
\draw[->,>=stealth', color=blue](1.196,0.368)--(12.5,0.368);
\node at (11.3,-0.2) {\small{time}};
\node at (10.663,5.2) {\small{$g_n^1$}};
\draw (1.196,8.381)--(1.211,0.368)--(1.227,2.748)--(1.242,2.600)--(1.258,2.937)%
  --(1.273,2.836)--(1.288,2.983)--(1.304,3.161)--(1.319,3.283)--(1.334,3.389)--(1.350,3.441)%
  --(1.365,3.490)--(1.381,3.472)--(1.396,3.470)--(1.411,3.509)--(1.427,3.551)--(1.442,3.571)%
  --(1.457,3.621)--(1.473,3.666)--(1.488,3.672)--(1.504,3.709)--(1.519,3.724)--(1.534,3.755)%
  --(1.550,3.783)--(1.565,3.809)--(1.581,3.832)--(1.596,3.846)--(1.611,3.852)--(1.627,3.872)%
  --(1.642,3.890)--(1.657,3.906)--(1.673,3.922)--(1.688,3.937)--(1.704,3.950)--(1.719,3.963)%
  --(1.734,3.968)--(1.750,3.971)--(1.765,3.982)--(1.780,3.992)--(1.796,4.002)--(1.811,4.012)%
  --(1.827,4.021)--(1.842,4.030)--(1.857,4.038)--(1.873,4.046)--(1.888,4.053)--(1.904,4.060)%
  --(1.919,4.055)--(1.934,4.062)--(1.950,4.068)--(1.965,4.074)--(1.980,4.080)--(1.996,4.086)%
  --(2.011,4.092)--(2.027,4.091)--(2.042,4.096)--(2.057,4.101)--(2.073,4.106)--(2.088,4.110)%
  --(2.103,4.109)--(2.119,4.114)--(2.134,4.118)--(2.150,4.122)--(2.165,4.126)--(2.180,4.130)%
  --(2.196,4.134)--(2.211,4.138)--(2.226,4.141)--(2.242,4.145)--(2.257,4.144)--(2.273,4.147)%
  --(2.288,4.151)--(2.303,4.154)--(2.319,4.157)--(2.334,4.160)--(2.350,4.163)--(2.365,4.166)%
  --(2.380,4.168)--(2.396,4.171)--(2.411,4.162)--(2.426,4.165)--(2.442,4.167)--(2.457,4.164)%
  --(2.473,4.167)--(2.488,4.169)--(2.503,4.172)--(2.519,4.174)--(2.534,4.177)--(2.549,4.179)%
  --(2.565,4.168)--(2.580,4.171)--(2.596,4.173)--(2.611,4.175)--(2.626,4.177)--(2.642,4.179)%
  --(2.657,4.181)--(2.673,4.183)--(2.688,4.185)--(2.703,4.187)--(2.719,4.189)--(2.734,4.191)%
  --(2.749,4.193)--(2.765,4.194)--(2.780,4.196)--(2.796,4.198)--(2.811,4.200)--(2.826,4.201)%
  --(2.842,4.203)--(2.857,4.205)--(2.872,4.206)--(2.888,4.208)--(2.903,4.209)--(2.919,4.211)%
  --(2.934,4.212)--(2.949,4.214)--(2.965,4.215)--(2.980,4.216)--(2.996,4.218)--(3.011,4.219)%
  --(3.026,4.220)--(3.042,4.222)--(3.057,4.223)--(3.072,4.224)--(3.088,4.225)--(3.103,4.216)%
  --(3.119,4.211)--(3.134,4.213)--(3.149,4.214)--(3.165,4.215)--(3.180,4.216)--(3.195,4.218)%
  --(3.211,4.219)--(3.226,4.220)--(3.242,4.221)--(3.257,4.222)--(3.272,4.224)--(3.288,4.225)%
  --(3.303,4.221)--(3.319,4.222)--(3.334,4.223)--(3.349,4.224)--(3.365,4.225)--(3.380,4.226)%
  --(3.395,4.227)--(3.411,4.228)--(3.426,4.229)--(3.442,4.230)--(3.457,4.231)--(3.472,4.232)%
  --(3.488,4.233)--(3.503,4.234)--(3.518,4.235)--(3.534,4.236)--(3.549,4.237)--(3.565,4.238)%
  --(3.580,4.239)--(3.595,4.240)--(3.611,4.240)--(3.626,4.241)--(3.642,4.242)--(3.657,4.243)%
  --(3.672,4.244)--(3.688,4.245)--(3.703,4.245)--(3.718,4.244)--(3.734,4.245)--(3.749,4.245)%
  --(3.765,4.246)--(3.780,4.247)--(3.795,4.248)--(3.811,4.248)--(3.826,4.249)--(3.841,4.250)%
  --(3.857,4.251)--(3.872,4.251)--(3.888,4.252)--(3.903,4.252)--(3.918,4.253)--(3.934,4.254)%
  --(3.949,4.255)--(3.964,4.255)--(3.980,4.256)--(3.995,4.257)--(4.011,4.257)--(4.026,4.258)%
  --(4.041,4.258)--(4.057,4.259)--(4.072,4.260)--(4.088,4.260)--(4.103,4.261)--(4.118,4.262)%
  --(4.134,4.262)--(4.149,4.263)--(4.164,4.263)--(4.180,4.264)--(4.195,4.264)--(4.211,4.260)%
  --(4.226,4.261)--(4.241,4.261)--(4.257,4.262)--(4.272,4.262)--(4.287,4.263)--(4.303,4.263)%
  --(4.318,4.264)--(4.334,4.265)--(4.349,4.265)--(4.364,4.266)--(4.380,4.266)--(4.395,4.267)%
  --(4.411,4.267)--(4.426,4.268)--(4.441,4.268)--(4.457,4.269)--(4.472,4.269)--(4.487,4.270)%
  --(4.503,4.270)--(4.518,4.271)--(4.534,4.270)--(4.549,4.271)--(4.564,4.271)--(4.580,4.272)%
  --(4.595,4.272)--(4.610,4.272)--(4.626,4.273)--(4.641,4.273)--(4.657,4.274)--(4.672,4.274)%
  --(4.687,4.275)--(4.703,4.275)--(4.718,4.275)--(4.734,4.276)--(4.749,4.276)--(4.764,4.277)%
  --(4.780,4.277)--(4.795,4.278)--(4.810,4.278)--(4.826,4.278)--(4.841,4.279)--(4.857,4.279)%
  --(4.872,4.280)--(4.887,4.280)--(4.903,4.280)--(4.918,4.281)--(4.933,4.281)--(4.949,4.282)%
  --(4.964,4.282)--(4.980,4.282)--(4.995,4.283)--(5.010,4.283)--(5.026,4.283)--(5.041,4.284)%
  --(5.057,4.284)--(5.072,4.284)--(5.087,4.285)--(5.103,4.285)--(5.118,4.286)--(5.133,4.286)%
  --(5.149,4.286)--(5.164,4.287)--(5.180,4.287)--(5.195,4.287)--(5.210,4.288)--(5.226,4.288)%
  --(5.241,4.288)--(5.256,4.289)--(5.272,4.289)--(5.287,4.289)--(5.303,4.290)--(5.318,4.290)%
  --(5.333,4.290)--(5.349,4.291)--(5.364,4.291)--(5.380,4.291)--(5.395,4.291)--(5.410,4.292)%
  --(5.426,4.292)--(5.441,4.292)--(5.456,4.292)--(5.472,4.292)--(5.487,4.293)--(5.503,4.293)%
  --(5.518,4.293)--(5.533,4.294)--(5.549,4.294)--(5.564,4.294)--(5.579,4.295)--(5.595,4.295)%
  --(5.610,4.295)--(5.626,4.295)--(5.641,4.295)--(5.656,4.295)--(5.672,4.296)--(5.687,4.296)%
  --(5.702,4.296)--(5.718,4.296)--(5.733,4.297)--(5.749,4.297)--(5.764,4.297)--(5.779,4.297)%
  --(5.795,4.298)--(5.810,4.298)--(5.826,4.298)--(5.841,4.298)--(5.856,4.299)--(5.872,4.299)%
  --(5.887,4.299)--(5.902,4.299)--(5.918,4.300)--(5.933,4.300)--(5.949,4.300)--(5.964,4.300)%
  --(5.979,4.301)--(5.995,4.301)--(6.010,4.301)--(6.025,4.301)--(6.041,4.302)--(6.056,4.302)%
  --(6.072,4.302)--(6.087,4.302)--(6.102,4.303)--(6.118,4.303)--(6.133,4.303)--(6.149,4.303)%
  --(6.164,4.303)--(6.179,4.304)--(6.195,4.304)--(6.210,4.304)--(6.225,4.304)--(6.241,4.305)%
  --(6.256,4.305)--(6.272,4.305)--(6.287,4.305)--(6.302,4.305)--(6.318,4.305)--(6.333,4.305)%
  --(6.348,4.305)--(6.364,4.305)--(6.379,4.306)--(6.395,4.306)--(6.410,4.306)--(6.425,4.306)%
  --(6.441,4.307)--(6.456,4.307)--(6.472,4.307)--(6.487,4.307)--(6.502,4.307)--(6.518,4.307)%
  --(6.533,4.308)--(6.548,4.308)--(6.564,4.308)--(6.579,4.308)--(6.595,4.308)--(6.610,4.309)%
  --(6.625,4.309)--(6.641,4.309)--(6.656,4.309)--(6.671,4.309)--(6.687,4.310)--(6.702,4.310)%
  --(6.718,4.309)--(6.733,4.309)--(6.748,4.309)--(6.764,4.309)--(6.779,4.310)--(6.795,4.310)%
  --(6.810,4.310)--(6.825,4.310)--(6.841,4.310)--(6.856,4.310)--(6.871,4.311)--(6.887,4.311)%
  --(6.902,4.311)--(6.918,4.311)--(6.933,4.311)--(6.948,4.311)--(6.964,4.312)--(6.979,4.312)%
  --(6.994,4.312)--(7.010,4.312)--(7.025,4.312)--(7.041,4.312)--(7.056,4.313)--(7.071,4.313)%
  --(7.087,4.313)--(7.102,4.313)--(7.118,4.313)--(7.133,4.313)--(7.148,4.314)--(7.164,4.314)%
  --(7.179,4.314)--(7.194,4.314)--(7.210,4.314)--(7.225,4.314)--(7.241,4.315)--(7.256,4.315)%
  --(7.271,4.315)--(7.287,4.315)--(7.302,4.315)--(7.317,4.315)--(7.333,4.315)--(7.348,4.316)%
  --(7.364,4.316)--(7.379,4.316)--(7.394,4.316)--(7.410,4.316)--(7.425,4.316)--(7.441,4.316)%
  --(7.456,4.317)--(7.471,4.317)--(7.487,4.317)--(7.502,4.317)--(7.517,4.317)--(7.533,4.317)%
  --(7.548,4.317)--(7.564,4.318)--(7.579,4.318)--(7.594,4.318)--(7.610,4.318)--(7.625,4.318)%
  --(7.640,4.318)--(7.656,4.318)--(7.671,4.318)--(7.687,4.318)--(7.702,4.318)--(7.717,4.318)%
  --(7.733,4.318)--(7.748,4.319)--(7.763,4.319)--(7.779,4.319)--(7.794,4.319)--(7.810,4.319)%
  --(7.825,4.319)--(7.840,4.319)--(7.856,4.320)--(7.871,4.320)--(7.887,4.320)--(7.902,4.320)%
  --(7.917,4.320)--(7.933,4.320)--(7.948,4.320)--(7.963,4.320)--(7.979,4.321)--(7.994,4.321)%
  --(8.010,4.321)--(8.025,4.321)--(8.040,4.321)--(8.056,4.321)--(8.071,4.321)--(8.086,4.321)%
  --(8.102,4.321)--(8.117,4.322)--(8.133,4.322)--(8.148,4.322)--(8.163,4.322)--(8.179,4.322)%
  --(8.194,4.322)--(8.210,4.322)--(8.225,4.322)--(8.240,4.321)--(8.256,4.321)--(8.271,4.322)%
  --(8.286,4.320)--(8.302,4.320)--(8.317,4.320)--(8.333,4.320)--(8.348,4.321)--(8.363,4.321)%
  --(8.379,4.321)--(8.394,4.321)--(8.409,4.321)--(8.425,4.321)--(8.440,4.321)--(8.456,4.321)%
  --(8.471,4.321)--(8.486,4.321)--(8.502,4.321)--(8.517,4.321)--(8.533,4.322)--(8.548,4.322)%
  --(8.563,4.322)--(8.579,4.322)--(8.594,4.322)--(8.609,4.322)--(8.625,4.322)--(8.640,4.322)%
  --(8.656,4.322)--(8.671,4.323)--(8.686,4.323)--(8.702,4.323)--(8.717,4.323)--(8.732,4.323)%
  --(8.748,4.323)--(8.763,4.323)--(8.779,4.323)--(8.794,4.323)--(8.809,4.323)--(8.825,4.324)%
  --(8.840,4.324)--(8.856,4.324)--(8.871,4.324)--(8.886,4.324)--(8.902,4.324)--(8.917,4.324)%
  --(8.932,4.324)--(8.948,4.324)--(8.963,4.324)--(8.979,4.325)--(8.994,4.325)--(9.009,4.325)%
  --(9.025,4.325)--(9.040,4.325)--(9.055,4.325)--(9.071,4.325)--(9.086,4.325)--(9.102,4.325)%
  --(9.117,4.325)--(9.132,4.326)--(9.148,4.326)--(9.163,4.326)--(9.179,4.326)--(9.194,4.326)%
  --(9.209,4.326)--(9.225,4.326)--(9.240,4.326)--(9.255,4.326)--(9.271,4.326)--(9.286,4.326)%
  --(9.302,4.327)--(9.317,4.327)--(9.332,4.327)--(9.348,4.327)--(9.363,4.327)--(9.378,4.327)%
  --(9.394,4.327)--(9.409,4.327)--(9.425,4.327)--(9.440,4.327)--(9.455,4.327)--(9.471,4.328)%
  --(9.486,4.328)--(9.501,4.328)--(9.517,4.328)--(9.532,4.328)--(9.548,4.328)--(9.563,4.328)%
  --(9.578,4.328)--(9.594,4.328)--(9.609,4.328)--(9.625,4.328)--(9.640,4.329)--(9.655,4.329)%
  --(9.671,4.329)--(9.686,4.329)--(9.701,4.329)--(9.717,4.329)--(9.732,4.329)--(9.748,4.329)%
  --(9.763,4.329)--(9.778,4.329)--(9.794,4.329)--(9.809,4.329)--(9.824,4.329)--(9.840,4.329)%
  --(9.855,4.329)--(9.871,4.329)--(9.886,4.329)--(9.901,4.329)--(9.917,4.329)--(9.932,4.327)%
  --(9.948,4.327)--(9.963,4.327)--(9.978,4.327)--(9.994,4.327)--(10.009,4.327)--(10.024,4.327)%
  --(10.040,4.327)--(10.055,4.327)--(10.071,4.327)--(10.086,4.327)--(10.101,4.328)--(10.117,4.328)%
  --(10.132,4.328)--(10.147,4.328)--(10.163,4.328)--(10.178,4.328)--(10.194,4.328)--(10.209,4.328)%
  --(10.224,4.328)--(10.240,4.328)--(10.255,4.328)--(10.271,4.328)--(10.286,4.328)--(10.301,4.329)%
  --(10.317,4.329)--(10.332,4.329)--(10.347,4.329)--(10.363,4.329)--(10.378,4.329)--(10.394,4.329)%
  --(10.409,4.329)--(10.424,4.329)--(10.440,4.329)--(10.455,4.329)--(10.470,4.329)--(10.486,4.329)%
  --(10.501,4.329)--(10.517,4.329)--(10.532,4.329)--(10.547,4.329)--(10.563,4.329)--(10.578,4.327)%
  --(10.594,4.328)--(10.609,4.328)--(10.624,4.328)--(10.640,4.327)--(10.655,4.327)--(10.670,4.327)%
  --(10.686,4.327)--(10.701,4.327)--(10.717,4.327)--(10.732,4.327)--(10.747,4.327)--(10.763,4.327)%
  --(10.778,4.327)--(10.793,4.327)--(10.809,4.328)--(10.824,4.328)--(10.840,4.328)--(10.855,4.328)%
  --(10.870,4.328)--(10.886,4.328)--(10.901,4.328)--(10.917,4.328)--(10.932,4.328)--(10.947,4.328)%
  --(10.963,4.328)--(10.978,4.328)--(10.993,4.328)--(11.009,4.328)--(11.024,4.329)--(11.040,4.329)%
  --(11.055,4.329)--(11.070,4.329)--(11.086,4.329)--(11.101,4.329)--(11.116,4.329)--(11.132,4.329)%
  --(11.147,4.329)--(11.163,4.329)--(11.178,4.329)--(11.193,4.329)--(11.209,4.329)--(11.224,4.329)%
  --(11.239,4.329)--(11.255,4.330)--(11.270,4.330)--(11.286,4.330)--(11.301,4.330)--(11.316,4.330)%
  --(11.332,4.330)--(11.347,4.330)--(11.363,4.330)--(11.378,4.330)--(11.393,4.330)--(11.409,4.330)%
  --(11.424,4.330)--(11.439,4.330)--(11.455,4.330)--(11.470,4.330)--(11.486,4.331)--(11.501,4.331)%
  --(11.516,4.331)--(11.532,4.331)--(11.547,4.331)--(11.562,4.331)--(11.578,4.331)--(11.593,4.331)%
  --(11.609,4.331)--(11.624,4.331)--(11.639,4.331)--(11.655,4.331)--(11.670,4.331)--(11.686,4.331)%
  --(11.701,4.331)--(11.716,4.332)--(11.732,4.332)--(11.747,4.332)--(11.762,4.332)--(11.778,4.332)%
  --(11.793,4.332)--(11.809,4.332)--(11.824,4.332)--(11.839,4.332)--(11.855,4.332)--(11.870,4.332)%
  --(11.885,4.332)--(11.901,4.332)--(11.916,4.332)--(11.932,4.332)--(11.947,4.332);
\end{tikzpicture}
  }
\hspace{2cm}
 \subfloat[On top, $v_n^1(B) \to 1$. On the bottom, $\overline g_n^1 \to 2$.]{  
\begin{tikzpicture}[scale=0.5]
\draw (1.196,0.368)--(1.376,0.368);
\node at (0.5,0.368) {\small{0}};
\draw (1.196,2.371)--(1.376,2.371);
\node at (0.5,2.371) {\small{ 0.5}};
\draw (1.196,4.375)--(1.376,4.375);
\node at (0.5,4.375) {\small{ 1}};
\draw (1.196,6.378)--(1.376,6.378);
\node at (0.5,6.378) {\small{ 1.5}};
\node at (0.5,8.381) {\small{2}};
\draw[->,>=stealth', color=blue](1.196,0.368)-- (1.196,8.381);
\draw[->,>=stealth', color=blue](1.196,0.368)--(12.5,0.368);
\node at (11.3,-0.2) {\small{time}};
\node at (10.663,7.5) {$g_n^1$};
\draw (1.196,1.241)--(1.211,3.198)--(1.227,4.920)--(1.242,5.783)--(1.258,6.302)%
  --(1.273,6.648)--(1.288,6.895)--(1.304,7.081)--(1.319,7.225)--(1.334,7.340)--(1.350,7.428)%
  --(1.365,7.367)--(1.381,7.445)--(1.396,7.512)--(1.411,7.570)--(1.427,7.620)--(1.442,7.665)%
  --(1.457,7.705)--(1.473,7.740)--(1.488,7.772)--(1.504,7.801)--(1.519,7.828)--(1.534,7.852)%
  --(1.550,7.874)--(1.565,7.894)--(1.581,7.913)--(1.596,7.930)--(1.611,7.946)--(1.627,7.961)%
  --(1.642,7.975)--(1.657,7.988)--(1.673,8.001)--(1.688,8.012)--(1.704,8.023)--(1.719,8.033)%
  --(1.734,8.043)--(1.750,8.052)--(1.765,8.061)--(1.780,8.069)--(1.796,8.077)--(1.811,8.084)%
  --(1.827,8.091)--(1.842,8.098)--(1.857,8.104)--(1.873,8.110)--(1.888,8.116)--(1.904,8.122)%
  --(1.919,8.127)--(1.934,8.132)--(1.950,8.137)--(1.965,8.142)--(1.980,8.147)--(1.996,8.151)%
  --(2.011,8.155)--(2.027,8.160)--(2.042,8.164)--(2.057,8.167)--(2.073,8.171)--(2.088,8.175)%
  --(2.103,8.178)--(2.119,8.181)--(2.134,8.185)--(2.150,8.188)--(2.165,8.191)--(2.180,8.194)%
  --(2.196,8.196)--(2.211,8.199)--(2.226,8.202)--(2.242,8.205)--(2.257,8.207)--(2.273,8.209)%
  --(2.288,8.212)--(2.303,8.214)--(2.319,8.216)--(2.334,8.219)--(2.350,8.221)--(2.365,8.223)%
  --(2.380,8.225)--(2.396,8.227)--(2.411,8.229)--(2.426,8.231)--(2.442,8.232)--(2.457,8.234)%
  --(2.473,8.236)--(2.488,8.238)--(2.503,8.239)--(2.519,8.241)--(2.534,8.243)--(2.549,8.244)%
  --(2.565,8.246)--(2.580,8.247)--(2.596,8.249)--(2.611,8.250)--(2.626,8.251)--(2.642,8.253)%
  --(2.657,8.254)--(2.673,8.255)--(2.688,8.257)--(2.703,8.258)--(2.719,8.259)--(2.734,8.260)%
  --(2.749,8.262)--(2.765,8.263)--(2.780,8.264)--(2.796,8.265)--(2.811,8.266)--(2.826,8.267)%
  --(2.842,8.268)--(2.857,8.269)--(2.872,8.270)--(2.888,8.271)--(2.903,8.272)--(2.919,8.273)%
  --(2.934,8.274)--(2.949,8.275)--(2.965,8.276)--(2.980,8.277)--(2.996,8.278)--(3.011,8.279)%
  --(3.026,8.280)--(3.042,8.280)--(3.057,8.281)--(3.072,8.282)--(3.088,8.283)--(3.103,8.284)%
  --(3.119,8.284)--(3.134,8.285)--(3.149,8.286)--(3.165,8.287)--(3.180,8.287)--(3.195,8.288)%
  --(3.211,8.289)--(3.226,8.289)--(3.242,8.290)--(3.257,8.291)--(3.272,8.291)--(3.288,8.292)%
  --(3.303,8.293)--(3.319,8.293)--(3.334,8.294)--(3.349,8.295)--(3.365,8.295)--(3.380,8.296)%
  --(3.395,8.296)--(3.411,8.297)--(3.426,8.298)--(3.442,8.298)--(3.457,8.299)--(3.472,8.299)%
  --(3.488,8.300)--(3.503,8.300)--(3.518,8.301)--(3.534,8.301)--(3.549,8.302)--(3.565,8.302)%
  --(3.580,8.303)--(3.595,8.303)--(3.611,8.304)--(3.626,8.304)--(3.642,8.305)--(3.657,8.305)%
  --(3.672,8.306)--(3.688,8.306)--(3.703,8.307)--(3.718,8.307)--(3.734,8.308)--(3.749,8.308)%
  --(3.765,8.309)--(3.780,8.309)--(3.795,8.309)--(3.811,8.310)--(3.826,8.310)--(3.841,8.311)%
  --(3.857,8.311)--(3.872,8.311)--(3.888,8.312)--(3.903,8.312)--(3.918,8.313)--(3.934,8.313)%
  --(3.949,8.313)--(3.964,8.314)--(3.980,8.314)--(3.995,8.314)--(4.011,8.315)--(4.026,8.315)%
  --(4.041,8.316)--(4.057,8.316)--(4.072,8.316)--(4.088,8.317)--(4.103,8.317)--(4.118,8.317)%
  --(4.134,8.318)--(4.149,8.318)--(4.164,8.318)--(4.180,8.319)--(4.195,8.319)--(4.211,8.319)%
  --(4.226,8.319)--(4.241,8.320)--(4.257,8.320)--(4.272,8.320)--(4.287,8.321)--(4.303,8.321)%
  --(4.318,8.321)--(4.334,8.322)--(4.349,8.322)--(4.364,8.322)--(4.380,8.322)--(4.395,8.323)%
  --(4.411,8.323)--(4.426,8.323)--(4.441,8.324)--(4.457,8.324)--(4.472,8.324)--(4.487,8.324)%
  --(4.503,8.325)--(4.518,8.325)--(4.534,8.325)--(4.549,8.325)--(4.564,8.326)--(4.580,8.326)%
  --(4.595,8.326)--(4.610,8.326)--(4.626,8.327)--(4.641,8.327)--(4.657,8.327)--(4.672,8.327)%
  --(4.687,8.328)--(4.703,8.328)--(4.718,8.328)--(4.734,8.328)--(4.749,8.329)--(4.764,8.329)%
  --(4.780,8.329)--(4.795,8.329)--(4.810,8.329)--(4.826,8.330)--(4.841,8.330)--(4.857,8.330)%
  --(4.872,8.330)--(4.887,8.330)--(4.903,8.331)--(4.918,8.331)--(4.933,8.331)--(4.949,8.331)%
  --(4.964,8.331)--(4.980,8.332)--(4.995,8.332)--(5.010,8.332)--(5.026,8.332)--(5.041,8.332)%
  --(5.057,8.333)--(5.072,8.333)--(5.087,8.333)--(5.103,8.333)--(5.118,8.333)--(5.133,8.334)%
  --(5.149,8.334)--(5.164,8.334)--(5.180,8.334)--(5.195,8.334)--(5.210,8.335)--(5.226,8.335)%
  --(5.241,8.335)--(5.256,8.335)--(5.272,8.335)--(5.287,8.335)--(5.303,8.336)--(5.318,8.336)%
  --(5.333,8.336)--(5.349,8.336)--(5.364,8.336)--(5.380,8.336)--(5.395,8.337)--(5.410,8.337)%
  --(5.426,8.337)--(5.441,8.337)--(5.456,8.337)--(5.472,8.337)--(5.487,8.338)--(5.503,8.338)%
  --(5.518,8.338)--(5.533,8.338)--(5.549,8.338)--(5.564,8.338)--(5.579,8.338)--(5.595,8.339)%
  --(5.610,8.339)--(5.626,8.339)--(5.641,8.339)--(5.656,8.339)--(5.672,8.339)--(5.687,8.339)%
  --(5.702,8.340)--(5.718,8.340)--(5.733,8.340)--(5.749,8.340)--(5.764,8.340)--(5.779,8.340)%
  --(5.795,8.340)--(5.810,8.341)--(5.826,8.341)--(5.841,8.341)--(5.856,8.341)--(5.872,8.341)%
  --(5.887,8.341)--(5.902,8.341)--(5.918,8.341)--(5.933,8.342)--(5.949,8.342)--(5.964,8.342)%
  --(5.979,8.342)--(5.995,8.342)--(6.010,8.342)--(6.025,8.342)--(6.041,8.342)--(6.056,8.343)%
  --(6.072,8.343)--(6.087,8.343)--(6.102,8.343)--(6.118,8.343)--(6.133,8.343)--(6.149,8.343)%
  --(6.164,8.343)--(6.179,8.344)--(6.195,8.344)--(6.210,8.344)--(6.225,8.344)--(6.241,8.344)%
  --(6.256,8.344)--(6.272,8.344)--(6.287,8.344)--(6.302,8.344)--(6.318,8.345)--(6.333,8.345)%
  --(6.348,8.345)--(6.364,8.345)--(6.379,8.345)--(6.395,8.345)--(6.410,8.345)--(6.425,8.345)%
  --(6.441,8.345)--(6.456,8.345)--(6.472,8.346)--(6.487,8.346)--(6.502,8.346)--(6.518,8.346)%
  --(6.533,8.346)--(6.548,8.346)--(6.564,8.346)--(6.579,8.346)--(6.595,8.346)--(6.610,8.346)%
  --(6.625,8.347)--(6.641,8.347)--(6.656,8.347)--(6.671,8.347)--(6.687,8.347)--(6.702,8.347)%
  --(6.718,8.347)--(6.733,8.347)--(6.748,8.347)--(6.764,8.347)--(6.779,8.348)--(6.795,8.348)%
  --(6.810,8.348)--(6.825,8.348)--(6.841,8.348)--(6.856,8.348)--(6.871,8.348)--(6.887,8.348)%
  --(6.902,8.348)--(6.918,8.348)--(6.933,8.348)--(6.948,8.349)--(6.964,8.349)--(6.979,8.349)%
  --(6.994,8.349)--(7.010,8.349)--(7.025,8.349)--(7.041,8.349)--(7.056,8.349)--(7.071,8.349)%
  --(7.087,8.349)--(7.102,8.349)--(7.118,8.349)--(7.133,8.350)--(7.148,8.350)--(7.164,8.350)%
  --(7.179,8.350)--(7.194,8.350)--(7.210,8.350)--(7.225,8.350)--(7.241,8.350)--(7.256,8.350)%
  --(7.271,8.350)--(7.287,8.350)--(7.302,8.350)--(7.317,8.350)--(7.333,8.351)--(7.348,8.351)%
  --(7.364,8.351)--(7.379,8.351)--(7.394,8.351)--(7.410,8.351)--(7.425,8.351)--(7.441,8.351)%
  --(7.456,8.351)--(7.471,8.351)--(7.487,8.351)--(7.502,8.351)--(7.517,8.351)--(7.533,8.352)%
  --(7.548,8.352)--(7.564,8.352)--(7.579,8.352)--(7.594,8.352)--(7.610,8.352)--(7.625,8.352)%
  --(7.640,8.352)--(7.656,8.352)--(7.671,8.352)--(7.687,8.352)--(7.702,8.352)--(7.717,8.352)%
  --(7.733,8.352)--(7.748,8.352)--(7.763,8.353)--(7.779,8.353)--(7.794,8.353)--(7.810,8.353)%
  --(7.825,8.353)--(7.840,8.353)--(7.856,8.353)--(7.871,8.353)--(7.887,8.353)--(7.902,8.353)%
  --(7.917,8.353)--(7.933,8.353)--(7.948,8.353)--(7.963,8.353)--(7.979,8.353)--(7.994,8.354)%
  --(8.010,8.354)--(8.025,8.354)--(8.040,8.354)--(8.056,8.354)--(8.071,8.354)--(8.086,8.354)%
  --(8.102,8.354)--(8.117,8.354)--(8.133,8.354)--(8.148,8.354)--(8.163,8.354)--(8.179,8.354)%
  --(8.194,8.354)--(8.210,8.354)--(8.225,8.354)--(8.240,8.354)--(8.256,8.355)--(8.271,8.355)%
  --(8.286,8.355)--(8.302,8.355)--(8.317,8.355)--(8.333,8.355)--(8.348,8.355)--(8.363,8.355)%
  --(8.379,8.355)--(8.394,8.355)--(8.409,8.355)--(8.425,8.355)--(8.440,8.355)--(8.456,8.355)%
  --(8.471,8.355)--(8.486,8.355)--(8.502,8.355)--(8.517,8.355)--(8.533,8.356)--(8.548,8.356)%
  --(8.563,8.356)--(8.579,8.356)--(8.594,8.356)--(8.609,8.356)--(8.625,8.356)--(8.640,8.356)%
  --(8.656,8.356)--(8.671,8.356)--(8.686,8.356)--(8.702,8.356)--(8.717,8.356)--(8.732,8.356)%
  --(8.748,8.356)--(8.763,8.356)--(8.779,8.356)--(8.794,8.356)--(8.809,8.356)--(8.825,8.356)%
  --(8.840,8.357)--(8.856,8.357)--(8.871,8.357)--(8.886,8.357)--(8.902,8.357)--(8.917,8.357)%
  --(8.932,8.357)--(8.948,8.357)--(8.963,8.357)--(8.979,8.357)--(8.994,8.357)--(9.009,8.357)%
  --(9.025,8.357)--(9.040,8.357)--(9.055,8.357)--(9.071,8.357)--(9.086,8.357)--(9.102,8.357)%
  --(9.117,8.357)--(9.132,8.357)--(9.148,8.357)--(9.163,8.358)--(9.179,8.358)--(9.194,8.358)%
  --(9.209,8.358)--(9.225,8.358)--(9.240,8.358)--(9.255,8.358)--(9.271,8.358)--(9.286,8.358)%
  --(9.302,8.358)--(9.317,8.358)--(9.332,8.358)--(9.348,8.358)--(9.363,8.358)--(9.378,8.358)%
  --(9.394,8.358)--(9.409,8.358)--(9.425,8.358)--(9.440,8.358)--(9.455,8.358)--(9.471,8.358)%
  --(9.486,8.358)--(9.501,8.358)--(9.517,8.359)--(9.532,8.359)--(9.548,8.359)--(9.563,8.359)%
  --(9.578,8.359)--(9.594,8.359)--(9.609,8.359)--(9.625,8.359)--(9.640,8.359)--(9.655,8.359)%
  --(9.671,8.359)--(9.686,8.359)--(9.701,8.359)--(9.717,8.359)--(9.732,8.359)--(9.748,8.359)%
  --(9.763,8.359)--(9.778,8.359)--(9.794,8.359)--(9.809,8.359)--(9.824,8.359)--(9.840,8.359)%
  --(9.855,8.359)--(9.871,8.359)--(9.886,8.359)--(9.901,8.360)--(9.917,8.360)--(9.932,8.360)%
  --(9.948,8.360)--(9.963,8.360)--(9.978,8.360)--(9.994,8.360)--(10.009,8.360)--(10.024,8.360)%
  --(10.040,8.360)--(10.055,8.360)--(10.071,8.360)--(10.086,8.360)--(10.101,8.360)--(10.117,8.360)%
  --(10.132,8.360)--(10.147,8.360)--(10.163,8.360)--(10.178,8.360)--(10.194,8.360)--(10.209,8.360)%
  --(10.224,8.360)--(10.240,8.360)--(10.255,8.360)--(10.271,8.360)--(10.286,8.360)--(10.301,8.360)%
  --(10.317,8.360)--(10.332,8.361)--(10.347,8.361)--(10.363,8.361)--(10.378,8.361)--(10.394,8.361)%
  --(10.409,8.361)--(10.424,8.361)--(10.440,8.361)--(10.455,8.361)--(10.470,8.361)--(10.486,8.361)%
  --(10.501,8.361)--(10.517,8.361)--(10.532,8.361)--(10.547,8.361)--(10.563,8.361)--(10.578,8.361)%
  --(10.594,8.361)--(10.609,8.361)--(10.624,8.361)--(10.640,8.361)--(10.655,8.361)--(10.670,8.361)%
  --(10.686,8.361)--(10.701,8.361)--(10.717,8.361)--(10.732,8.361)--(10.747,8.361)--(10.763,8.361)%
  --(10.778,8.361)--(10.793,8.362)--(10.809,8.362)--(10.824,8.362)--(10.840,8.362)--(10.855,8.362)%
  --(10.870,8.362)--(10.886,8.362)--(10.901,8.362)--(10.917,8.362)--(10.932,8.362)--(10.947,8.362)%
  --(10.963,8.362)--(10.978,8.362)--(10.993,8.362)--(11.009,8.362)--(11.024,8.362)--(11.040,8.362)%
  --(11.055,8.362)--(11.070,8.362)--(11.086,8.362)--(11.101,8.362)--(11.116,8.362)--(11.132,8.362)%
  --(11.147,8.362)--(11.163,8.362)--(11.178,8.362)--(11.193,8.362)--(11.209,8.362)--(11.224,8.362)%
  --(11.239,8.362)--(11.255,8.362)--(11.270,8.362)--(11.286,8.362)--(11.301,8.362)--(11.316,8.363)%
  --(11.332,8.363)--(11.347,8.363)--(11.363,8.363)--(11.378,8.363)--(11.393,8.363)--(11.409,8.363)%
  --(11.424,8.363)--(11.439,8.363)--(11.455,8.363)--(11.470,8.363)--(11.486,8.363)--(11.501,8.363)%
  --(11.516,8.363)--(11.532,8.363)--(11.547,8.363)--(11.562,8.363)--(11.578,8.363)--(11.593,8.363)%
  --(11.609,8.363)--(11.624,8.363)--(11.639,8.363)--(11.655,8.363)--(11.670,8.363)--(11.686,8.363)%
  --(11.701,8.363)--(11.716,8.363)--(11.732,8.363)--(11.747,8.363)--(11.762,8.363)--(11.778,8.363)%
  --(11.793,8.363)--(11.809,8.363)--(11.824,8.363)--(11.839,8.363)--(11.855,8.363)--(11.870,8.363)%
  --(11.885,8.364)--(11.901,8.364)--(11.916,8.364)--(11.932,8.364)--(11.947,8.364);
\end{tikzpicture}
 }
\end{center}
\caption[lineheight]{\small{Two realizations of the procedure for the game \eqref{coordination}. }}
\label{fig_coordination}
\end{figure}
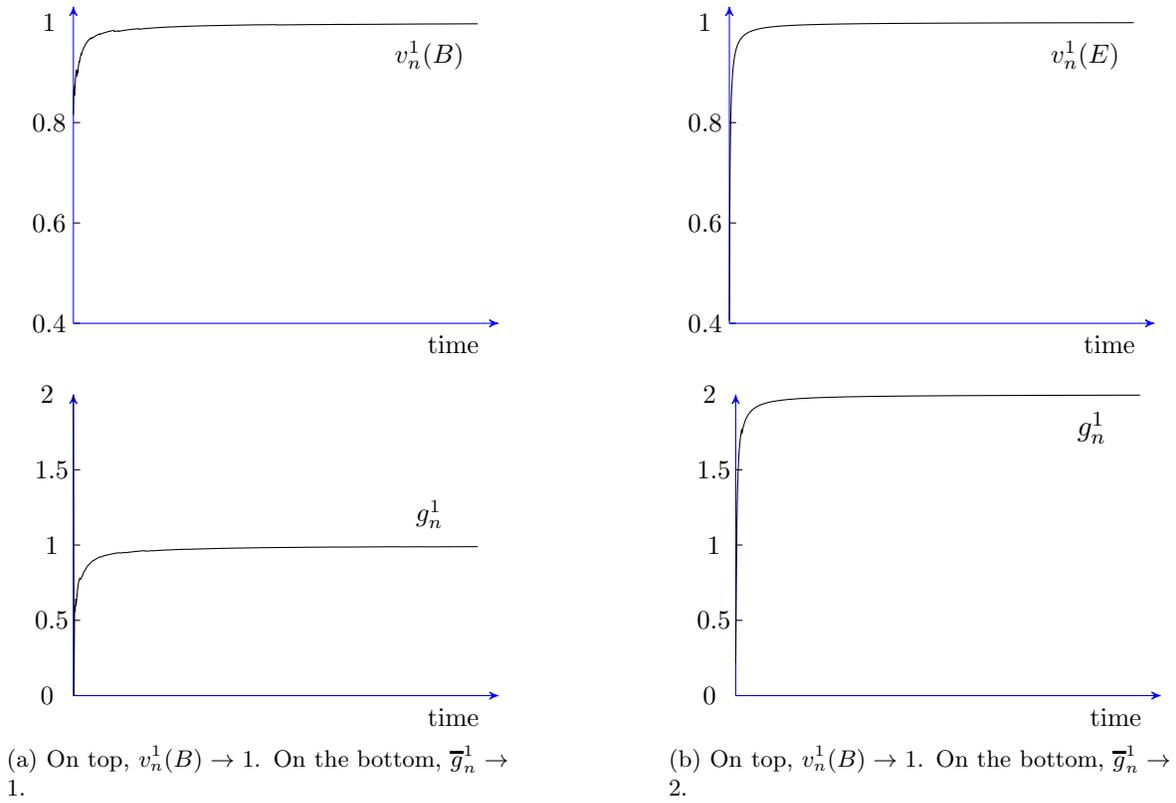

\section{Preliminaries to the proof, related work}\label{sec:preliminaries}
The aim of this section is twofold: we introduce the general framework in which we analyze our procedure, and we  present the related \textit{Markovian fictitious play} procedure, where the idea of restrictions on the action set was first introduced.

\subsection{A general framework}  \label{ss:general}

\noindent Let $S$ be a finite set and let $\mathcal M(S)$ be the set of Markov matrices over $S$. We consider a discrete time stochastic process $(s_n,M_n)_n$ defined on the probability space $(\Omega, \mathcal F, \mathbb P)$ and taking values in $S \times \mathcal M(S)$. The  space $(\Omega, \mathcal F, \mathbb P)$ is equipped with a non-decreasing sequence of $\sigma$-algebras $(\mathcal F_n)_n$.

\noindent Let us assume that the process  $(s_n, M_n)_n$ satisfies the following assumptions. 
\begin{equation}
 \begin{aligned}
 \text{(i) }& \text{The sequence } (s_n,M_n)_n \text{ is } \mathcal F_n\text{-measurable for all } n \in \NN.\\
 \text{(ii) } & \text{For all } s \in S, \, \pcon{s_{n+1} = s}{ \FF_n}= M_n(s_n,s). \\
  \text{(iii) } & \text{For all } n \in \NN, \text{ the matrix } M_n \text{ is irreducible with invariant measure } \, \pi_n \in \Delta(S). 
\end{aligned} \label{hip0}\tag{$H_0$}
\end{equation}

\noindent Let $\Sigma$ be a compact convex subset of an euclidean space and, for all $n \in \NN$, let $V_n=H(s_n) \in \Sigma$. We are interested in the asymptotic behavior of the random sequence $z_n= \frac{1}{n} \sum_{m=1}^n V_m$. Let us call 
\begin{equation}
\mu_n = \sum_{s \in S} \pi_n(s) H(s) \in \Sigma. \label{eq:mu}   
\end{equation}
\begin{remark}
 This setting is a simplification of that considered by Bena\"im and Raimond~\cite{br10}, where a of more general observation term $V_n$ is treated. For instance $V_n$ may depend on other non-observable variables or explicitly on time.
\end{remark}
\noindent In order to maintain the terminology employed in \cite{br10}, we introduce the following definition, which is stated in a slightly different form (see \cite[Definition~2.4]{br10}).  
\begin{definition}
A set-valued map  with nonempty convex values $C: \;  \Sigma \rightrightarrows \Sigma$ is adapted to the random sequence $(z_n,\mu_n)_n$ if 
\begin{itemize}
\item[i)] its graph
\[\Gr(C) = \left\{ (z,\mu): \; \, z \in \Sigma, \; \mu \in C(z) \right\}\]
is closed in $\Sigma \times \Sigma$.
\item[ii)]Almost surely, for any limit point $(z,\mu)$ of $(z_n,\mu_n)_n$, we have $(z,\mu) \in Gr(C)$.
\end{itemize}
 
\end{definition}

\noindent Given a set-valued map $C: \;  \Sigma \rightrightarrows \Sigma$ adapted to a random sequence $(z_n,\mu_n)_n$, we consider the differential inclusion
\begin{equation}
\dot{z} \in -z + C(z).  \label{general_continuo} \tag{DI}
\end{equation}

\noindent Under the assumptions above, it is well known (see, \emph{e.g.} Aubin and Cellina \cite{AubCel84}) that (\ref{general_continuo}) admits at least one solution (i.e. an absolutely continuous mapping $\mathbf{z} : \mathbb{R} \rightarrow \mathbb{R}^d$ such that $\dot{\mathbf{z}}(t) \in -\mathbf{z}(t) + C(\mathbf{z}(t))$ for almost every $t$) through any initial point.

\begin{definition}
A nonempty compact set $\mathcal A \subseteq \Sigma$ is called an \emph{attractor} for \eqref{general_continuo}, provided
\begin{itemize}
\item[$(i)$]  it is \emph{invariant}, i.e. for all $v\in \mathcal A$, there exists a solution $\mathbf{z}$ to \eqref{general_continuo} with $\mathbf{z}(0) =v$ and such that $\mathbf{z}(\mathbb{R}) \subseteq \mathcal A$,
\item[$(ii)$] there exists an open neighborhood $\mathcal U$ of $\mathcal A$ such that, for every $\epsilon >0$, there exists $t_{\epsilon}>0	$ such that
$$\mathbf{z}(t) \subseteq N^{\epsilon}(\mathcal A)$$
for any solution $\mathbf{z}$ starting in $\mathcal U$ and all $t>t_{\epsilon}$, where $N^{\varepsilon}(\mathcal A)$ is the $\varepsilon$-neighborhood of $\mathcal A$. An open set $\mathcal U$ with this property is called a \emph{fundamental neighborhood} of $\mathcal A$.
\end{itemize}
\end{definition}

\noindent A compact set $D \subseteq \Sigma$ is \emph{internally chain transitive} (ICT) if it is invariant, connected and has no proper attractors. See Bena\"im, Hofbauer and Sorin~\cite{bhs05} for a full discussion.

\noindent Let $m(t)=\sup \{ m \geq 0:\, t \geq \tau_m \}$, where $\tau_m= \sum_{j=1}^{m}\frac{1}{j}$.  For a sequence $(u_n)_n$ and a number $T>0$, we define $\epsilon(u_n,T)$ by
\begin{equation*}
\epsilon(u_n,T)= \sup  \left\{ \norm{\sum \limits_{j=n}^{l-1} u_{j+1}}; \; \, l \in \{ n+1, \dots, m(\tau_n + T) \} \right\}.
\end{equation*}

\noindent Let us denote by $(W_n)_n$ the random sequence defined by $W_{n+1} = H(s_{n+1}) - \mu_n$. The evolution of $z_n$ can be recast as
\begin{equation}
  \label{eq:v}
  z_{n+1} - z_n= \dfrac{1}{n+1}(\mu_n - z_n + W_{n+1} ).
\end{equation}

\noindent  A consequence of \cite[Theorem 2.6]{br10} in this particular framework is the following result.
\begin{theorem} \label{general}
 Assume that the set-valued map $C$ is adapted to $(z_n,\mu_n)_n$ and that for all $T>0$
\begin{equation} \label{noise_general}
\lim_{n \to +\infty} \epsilon \left(\frac{1}{n+1} W_{n+1},T \right) = 0,
 \end{equation}
almost surely. Then the limit set of $(z_n)_n$ is, almost surely, an \textup{ICT} set of the differential inclusion \eqref{general_continuo}. In particular, if $\mathcal A$ is a global attractor for \eqref{general_continuo} then the limit set of $(z_n)_n$ is almost surely contained in $\mathcal A$.
\end{theorem}
\begin{remark} Roughly speaking, the fact that the set-valued map $C$ is adapted to $(z_n,\mu_n)$ means that  \eqref{eq:v} can be recast as
\begin{equation*}
  \label{eq:v'}
  z_{n+1} - z_n \in \dfrac{1}{n+1}( - z_n + C(z_n) + W_{n+1} ).
\end{equation*}
In turn, this recursive form can be seen as a Cauchy-Euler Scheme to approximate the solutions of the differential inclusion  \eqref{general_continuo} with decreasing step sizes and  added noise term $(W_n)_n$. Assumption~\eqref{noise_general} guarantees that, on any given  time horizon,  the noise term asymptotically vanishes. As a consequence, the limit set of $(z_n)_n$ can be described through the deterministic dynamics \eqref{general_continuo}, in the sense that it needs to be internally chain transitive. If the differential inclusion admits a global attractor, then any ICT set is contained in it. This implies the second point of the theorem (again, see \cite{bhs05} for more details about stochastic approximations for differential inclusions) .
\end{remark}

\subsection{Markovian fictitious play}\label{sec:mark-fict-play}

\noindent Bena\"im and Raimond~\cite{br10} introduce an adaptive process they call \textit{Markovian fictitious play} (MFP). As in Section~\ref{sec:model}, we consider that players have constraints on their action set, i.e. each player has an exploration matrix $M_0^i$ which is supposed to be irreducible and reversible with respect to its unique invariant measure $\pi_0^i$.

\noindent The crucial difference between (MFP) and the procedure introduced in Section~\ref{sec:MFP} is that players know their own payoff function. Also, at the end of each stage, each player is informed of the opponent's action. The (MFP) procedure is defined as follows. A player's i action at time $n+1$ is chosen accordingly to the non-homogeneous Markov matrix 
\begin{equation} \tag{MFP} \label{eq:M_n}
  \begin{aligned}
 \mathbb P( s_{n+1}^i=s \, \vert \, \mathcal F_n)&= M^i[\beta_n^i, U_n^i](s_n^i,s),\\
&= \begin{cases} M_0^i(s_n^1,s) \expo(-\beta_n^i  |U_n^i(s_n^i)- U_n^i(s)|_+) &s \neq s_n^1, \\ 1 - \sum \limits_{s' \neq s}M^i[\beta_n^i, U_n^i](s_n^i,s') & s =s_n^1, \\ \end{cases} 
  \end{aligned}
\end{equation}
 where $U_n^i$ is taken as the vector payoffs of player $i$, against the average moves of the opponent
\begin{equation}
  \label{eq:average_U}
  U_n^i =G^i(\cdot, v_n^{-i})= \dfrac{1}{n} \sum_{m=1}^n G^i(\cdot, s_m^{-i}),
\end{equation}
\noindent for all $s \in S^i$, and the function $M^i[\cdot, \cdot]$ is defined by \eqref{def_M}. Let $\tilde M_n^i=M^i[\beta_n^i, U_n^i]$.
\noindent Observe that again, from the irreducibility of $M_0^i$, the matrix $\tilde M_n^i$ is also irreducible. Also,  $\tilde \pi_n^i=\pi^i[\beta_n^i,G^i(\cdot, v^{-i}_n)]$ (where  $\pi^i[\cdot, \cdot]$ is defined in \eqref{mesinv}) is the unique invariant measure of $\tilde M_n^i$, i.e.
\begin{equation}\label{eq:inv}
\tilde \pi_n(s)= \frac{\pi_0^i(s)\expo(\beta_n^i U_n^i(s))}{\sum \limits_{s' \in S^i}\pi_0^i(s')\expo(\beta_n^i U_n^i(s'))},
\end{equation}
for every $s \in S^i$.

\noindent Bena\"im and Raimond \cite{br10} obtain the following result.

\begin{theorem} \label{th:BR10}
If both players follow the \textup{(MFP)} procedure, then the limit set of the sequence  $v_n=(v^1_n, v_n^2)$  is an {\em ICT} set of the Best-Response dynamics \eqref{BR}, provided for $i \in \{1,2\}$ the positive sequence $(\beta_n^i)_n$ satisfies
\begin{enumerate}
\item  $\beta_n^i \to +\infty$ as $n \to +\infty$.
\item  $\beta_n^i \leq A^i \log (n)$, for a sufficiently small positive constant $A^i$.
\end{enumerate}
 As a consequence, we have the following.
 \begin{enumerate}
 \item[\textup{(a)}] In a zero-sum game, $(v_n^1, v_n^2)_n$ converges almost surely to the set of Nash equilibria. 
\item[\textup{(b)}] If $ G^1=G^2$,  then  $(v_n^1, v_n^2)_n$ converges almost surely to a connected subset of the set of Nash equilibria on which $G^1$ is constant.
 \end{enumerate}
\end{theorem}

\subsubsection*{Some insights on the proof of Theorem \ref{th:BR10}}
\noindent We believe it is interesting to sketch the proof of Theorem~\ref{th:BR10}. For that purpose, we need to introduce some notions that will be useful later on. 

\noindent Let $S$ be a finite set and $M$ an irreducible stochastic matrix over $S$ with invariant measure $\pi$. For a function $f : S \to \RR$, the variance and the energy of $f$ are defined, respectively, as
\begin{align*}
\var(f)&= \sum_{s \in S} \pi(s)f^2(s) - \left ( \sum_{s \in S} \pi(s)f(s)\right )^2,\\
\mathcal E(f,f) &= \frac{1}{2} \sum_{s,r \in S} (f(s) - f(r))^2M(s,r)\pi(s).
\end{align*}
 
\begin{definition}
Let $M$ be a stochastic irreducible matrix over the finite set $S$ and $\pi$ be its unique invariant measure.
\begin{enumerate}
\item The spectral gap of $M$ is defined by
\begin{equation}
\chi(M)= \min \biggl \{ \frac{\mathcal E(f,f)}{\var(f)} \, :\, \var(f)\neq 0 \biggr \}.
\end{equation}
\item The pseudo-inverse of $M$ is the unique matrix $Q \in \RR^{|S|\times|S|}$ such that $\sum_r Q(s,r) = 0$, for every $s \in S$, which  satisfies the Poisson's equation
\begin{equation}
 Q(I - M)= (I - M)Q= I- \Pi, \label{pseudo}
\end{equation}
where $\Pi$ is the matrix defined as $\Pi(s,r)=\pi(r)$ for every $s,r \in S$ and $I$ denotes the identity matrix.
\end{enumerate}
\end{definition}

\noindent For a matrix $Q \in \RR^{|S| \times |S|}$ and a vector $U \in \RR^{|S|}$, set $|Q|=\max_{s,r}|Q(s,r)|$ and $|U|= \max_s |U(s)|$.

\paragraph{Sketch of the proof of Theorem \ref{th:BR10}} We apply Theorem \ref{general} with $H(s) = (\delta_{s^1},\delta_{s^2})$. Recall that $v_n^i$ is the empirical frequency of play of player $i$. Thus, the random variable $z_n= v_n$ is given by
\[v_n = \frac{1}{n} \sum_{m=1}^n \left(\delta_{s_n^1}, \delta_{s_n^2} \right) = \left(v_n^1,v_n^2 \right).\]
Therefore, the evolution of $v_n$ is described by
\begin{equation}
    v_{n+1} - v_n= \dfrac{1}{n+1}(\mu_n - v_n + W_{n+1}),
\end{equation}
where $\mu_n = \sum_{s \in S}\tilde \pi_n(s) H(s) = \left(\tilde\pi_n^1,\tilde\pi_n^2 \right)$ and $\tilde W_{n+1}= (\delta_{s_n^1}, \delta_{s_n^2}) - \mu_n = \left(\delta_{s_n^1} - \tilde{\pi}_n^1, \delta_{s_n^2} - \tilde{\pi}_n^2 \right)$.

\noindent We first provide a sketch of the proof that \eqref{noise_general} holds for the sequence $(\tilde W_n)_n$. Afterwards, we will verify that the set-valued map $\BR$ is adapted to $(v_n,\mu_n)_n$ and conclude by applying Theorem~\ref{general}. 

\noindent Consequences $(a)$ and $(b)$ for games follow from the fact that the set of Nash equilibria is an attractor for the Best-Response dynamics in the relevant classes of games. We will omit this part of the proof, since the same argument will be used in Section~\ref{sec:proof-theor-refth:m}.

\noindent Let $\tilde Q_n^i$ be the pseudo-inverse of $\tilde M_n^i$. Bena\"im and Raimond prove that if , for $i \in \{1,2\}$,
\noindent 
\begin{equation}
 \begin{aligned}
   \lim \limits_{n \to +\infty} \frac{|\tilde Q^i_n|^2\ln(n)}{n}&=0, \\
    \lim \limits_{n \to +\infty} |\tilde Q^i_{n+1} - \tilde Q^i_n|&=0,\\
   \lim \limits_{n \to +\infty} |\tilde\pi^i_{n+1} - \tilde\pi^i_n|&=0, \ \
 \end{aligned} \label{H1}
\end{equation}
almost surely, then \eqref{noise_general} holds for $(\tilde W_n)_n$.

\noindent Proposition~3.4 in \cite{br10} shows that the norm of $\tilde Q_n^i$  can be controlled as a function of the spectral gap $\chi(\tilde M_n^i)$. If in addition the constants $A^i$ are sufficiently small, then \eqref{H1} holds. 

\noindent Finally, since $\beta_n^i \to +\infty$, we have that if $(v_n^1,v_n^2) \to (v^1,v^2)$, then $\tilde \pi_n^i \to \overline \pi^i[v^{-i}]$, where, for all $s \in S^i$
\begin{equation*}
 \overline \pi^i[v^{-i}](s)= \frac{\pi_0^i(s)\ind_{\{s \in \Argmax_r G^i(r, v^{-i}) \}}}{\sum \limits_{s' \in S^i}\pi_0^i(s')\ind_{\{s' \in \Argmax_r G^i(r, v^{-i}) \}}} \in \BR^i(v^{-i}).
\end{equation*}
\noindent This implies that map $\mathbf \BR$  is adapted to $(v_n, \mu_n)_n$ and the proof is finished.
\vspace{.3cm}
\section{Proof of the main result} \label{proofs}
As mentioned in Section~\ref{sec:mainresult}, we will prove a more general result. The following theorem implies that the conclusions of Theorem~\ref{th:BR10}  hold for our procedure.
\begin{theorem} \label{main2}
Under hypothesis \eqref{hyp_param}, assume that players follow the Payoff-based adaptive Markovian  procedure. Then the limit set of the sequence $(v_n)_n$ is an {\em ICT} set of the Best-Response dynamics \eqref{BR}. In particular if \eqref{BR} admits a global attractor $\mathcal A$, then $\mathcal L((v_n)_n) \subseteq \mathcal A$.
\end{theorem}
\noindent There are two key aspects which highlight the difference between the proof of Theorem~\ref{main2} and the proof of Theorem~\ref{th:BR10}. First, to show that the noise sequence (defined in \eqref{bruit:final} below) satisfies condition \eqref{noise_general}, we cannot directly use condition \eqref{H1}. Second, the proof that $\BR$ is adapted to $(v_n, \mu_n)_n$  is considerably more involved. In contrast to the approach for the (MFP) procedure, the invariant measure $\pi_n^i$ of matrix $M_n^i$ depends on state variable $R_n^i$ which is updated, in turn, using $\pi_{n-1}^i$. To overcome these difficulties, we develop a more general approach, that is presented in the Appendix.

\noindent In what follows, we present an extended sketch of the proof of Theorem~\ref{main2}. The proof of Theorem~\ref{th:main} will follow as a corollary.
\subsection{Proof of Theorem \ref{main2}}
\begin{proof}
We aim to apply Theorem~\ref{general}. Let $\Sigma=\Delta(S^1) \times \Delta(S^2)$. We take $V_n = (\delta_{s_n^1},\delta_{s_n^2})$ and  $\mu_n = (\pi_n^1,\pi_n^2)$. As before, let $v_n=(v_n^1, v_n^2)$. Then we have 
\begin{equation}
v_{n+1} - v_n = \frac{1}{n+1} \left ( \mu_n - v_n + \overline W_{n+1} \right ),
\end{equation}
where
\begin{equation}\label{bruit:final}
 \overline W_{n+1} = \left(\overline W_{n+1}^1,\overline W_{n+1}^2 \right) = (\delta_{s_{n+1}^1}-\pi_{n}^1, \delta_{s_{n+1}^2} - \pi_n^2) .
\end{equation}
We need to verify that two conditions hold. First we need to prove that $\epsilon \left(\frac{1}{n+1} \overline W_{n+1},T \right)$ goes to zero almost surely for all $T>0$.  Proposition~\ref{bruit}~(ii) provides  proof of this.

\noindent Second, we need to verify that the Best-Response correspondence $\BR$ is adapted to $\left (v_n, \mu_n \right )_n$.  As we will see, this problem basically amounts to showing that vector $R_n^i$ becomes a good asymptotic estimator of vector $G^i(\cdot, v_n^{-i})$.

\noindent Fix $i \in \{1,2\}$ and $s \in S^i$. Lemma~\ref{stabilite} shows that for a sufficiently large $n$, $(\gamma_{n+1}^i(s))^{-1}= (n+1)\pi_n^i(s)$ for any $s \in S^i$. Therefore, from the definition of $R_n^i$ and without any loss of generality, we have
\begin{align} 
R_{n+1}^i(s) - R_n^i(s) &= \frac{1}{(n+1)\pi_n^i(s)} \left [ \ind_{\{s_{n+1}^i=s \} }G^i(s, s_{n+1}^{-i}) - \ind_{\{s_{n+1}^i=s\}}R_n^i(s) \right ],\nonumber \\
&= \frac{1}{(n+1)\pi_n^i(s)} \left [  \pi_n^i(s) \left ( G^i(s, \pi_{n}^{-i})  - R_n^i(s) \right ) \right.  +  \phantom{\left ( \ind_{\{s_{n+1}^i=s \} }G^i(s, s_{n+1}^{-i})  - \pi_{n}^i(s) G^i(s, \pi^{-i}_{n})  \right ) }\nonumber \\
 &+\big ( \ind_{\{s_{n+1}^i=s \} }G^i(s, s_{n+1}^{-i})  - \pi_{n}^i(s) G^i(s, \pi^{-i}_{n})  \big )  \left. +  R_n^i(s)\left (  \pi_{n}^i(s)  - \ind_{\{s_{n+1}^i=s \}} \right)\right], \nonumber \\
 &= \frac{1}{n+1}\left [  G^i(s, \pi_{n}^{-i})  - R_n^i(s) + W_{n+1}^i(s) \right ],  \label{R_final}
\end{align}
where for convenience we set $W_{n+1}^i(s) = W_{n+1}^{i,1}(s) + W_{n+1}^{i,2}(s)$ with 
\begin{align}
W^{i,1}_{n+1}(s) &=\frac{R_n^i(s)}{\pi_{n}^i(s)} \left (\pi_{n}^i(s)  - \ind_{\{s_{n+1}^i=s \}} \right ), \text{    and   } \label{w1}\\
W^{i,2}_{n+1}(s) &=\frac{1}{\pi^i_{n}(s)} \left ( \ind_{\{s_{n+1}^i=s \}} G^i( s_{n+1}^1, s_{n+1}^2) -\pi_{n}^i(s) G^i(s, \pi^{-i}_{n}) \right ). \label{w2}
\end{align}

\noindent Propositions \ref{bruit}$(i)$ and \ref{bruit2} prove that, almost surely and for any $T>0$, $\epsilon \left(\frac{1}{n+1}  W_{n+1}^{i,1}(s),T \right) \to 0$ and $\epsilon \left(\frac{1}{n+1}  W_{n+1}^{i,2}(s),T \right) \to 0$, respectively.

\noindent Recall that $U_n^i = G^i(\cdot, v_n^{-i})$. Naturally, the evolution of vector $U_n^i$ can be written as
\begin{equation} \label{average}
U_{n+1}^i - U_{n}^i = \frac{1}{n+1} \left (G^i(\cdot,\pi_{n}^{-i}) - U_{n}^i + W^{i,3}_{n+1} \right ),\\
\end{equation}
where 
\begin{equation*}
W^{i,3}_{n+1} = G^i(\cdot ,s_{n+1}^{-i}) -G^i(\cdot,\pi_{n}^{-i}).  \label{bruit_m}
\end{equation*}
\noindent Again, Proposition~\ref{bruit}~(iii) shows that for all $T>0$, $\epsilon \left(\frac{1}{n+1}  W^{i,3}_{n+1},T \right) \to 0$ almost surely.

\noindent We define $\zeta_n^i=R_{n}^i - G^i(\cdot, v_{n}^{-i}) = R_{n}^i - U_{n}^i $.  Equations \eqref{R_final} and \eqref{average}, show that the evolution of the sequence $(\zeta_n^i)_n$ can be recast as 
\begin{equation} \label{average2}
\zeta^i_{n+1} - \zeta^i_n = \frac{1}{n+1} \big [ - \zeta^i_n + \mathcal W^i_{n+1}  \big ],
\end{equation}
where
\begin{equation*}
\mathcal W^i_{n+1}= W^{i,1}_{n+1}+W^{i,2}_{n+1}-W^{i,3}_{n+1},
\end{equation*}
and each component of $W^{i,1}_{n+1}$ and $W^{i,2}_{n+1}$ defined by Equations \eqref{w1} and \eqref{w2}, respectively.

\noindent Collecting all the analysis above, we conclude that $\epsilon \left(\frac{1}{n+1} \mathcal W^{i}_{n+1},T \right) \to 0$ almost surely for all $T>0$.

\noindent Based on the fact that sequence $(\zeta_n^i)_n$ is bounded (see Lemma~\ref{stabilite}) and on standard results from stochastic approximation theory  (see, e.g. Bena\"im~\cite{benaim99}), the limit set of the sequence $(\zeta_n^i)_n$ is almost surely an ICT set of the ordinary differential equation $\dot \zeta = -\zeta$ which admits the set $\{ 0\}$ as a global attractor.

\noindent Therefore, for $i \in \{1,2 \}$, $R_{n}^i- G^i(\cdot,v_{n}^{-i}) \to 0$ as $n \to +\infty$, almost surely.

\noindent Now let us assume that 
\begin{equation*}
(v_{n_{k}}^1,v_{n_{k}}^2)  \to (v^1,v^2) \in \Sigma, \text{ and }  (\pi_{n_{k}}^1, \pi_{n_{k}}^2) \to (\pi^1, \pi^2) \in \Sigma, 
\end{equation*}
for a sub-sequence $(n_k)_k$.

\noindent For $i \in \{1,2 \}$, let $r \notin \Argmax_{s'} G^i(s',v^{-i})$ and take $\hat{s} \in S^i$ such that 
\[G^i(r,v^{-i}) < G^i(\hat{s},v^{-i}).\]
 Since  $R_{n}^i- G^i(\cdot,v_{n}^{-i}) \to 0$, there exists $\varepsilon>0$ and $k_0 \in \NN$ such that, for any $k \geq k_0$,
\[R^i_{n_{k}}(r) < R^i_{n_{k}}(\hat{s})- \varepsilon.\]
So that, for $k$ sufficiently large,
\[\pi^i_{n_{k}}(r) \leq\frac{\pi_o^i(r)}{\pi_o^i(\hat s)} \exp \left [ \beta_{n_{k}}^i (R^i_{n_{k}}(r) - R^i_{n_{k}}(\hat{s})) \right ] \leq \frac{\pi_o^i(r)}{\pi_o^i(\hat s)} \exp (-\beta_{n_{k}}^i \varepsilon).\]
Then $\pi^i(r) = 0$ and we have proved that $\pi^i \in \BR^i(v^{-i})$ which implies that set-valued map $\BR$ is adapted to $\bigl (v_n, \mu_n\bigr )_n$.
\end{proof}

\subsection{Proof of Theorem~\ref{th:main}} \label{sec:proof-theor-refth:m}
For all three points, the result follows from a direct application of Theorem~\ref{main2}. 

\noindent Consider the variable $z_n= (v_n^1, v_n^2,\overline g_n^1, \overline g_n^2)$, where  $\overline g_n^i= \frac{1}{n}\sum_{m=1}^ng_m^i$ is the average realized payoff for player $i \in \{1,2\}$. Recall that the evolution of  $\overline g_n^i$ can be written as
\begin{equation*}
  \begin{aligned}
    \overline g_{n+1}^i -\overline g_{n+1}^i&= \frac{1}{n+1}\left ( g_{n+1}^i - \overline g_n^i\right )\\
&= \frac{1}{n+1}\left ( G^i(\pi_{n}^i,\pi_{n}^{-i})  - \overline g_n^i + W_{n+1}^{i,4}   \right ) ,
  \end{aligned}
\end{equation*}
where $ W_{n+1}^{i,4}=G^i(s_{n+1}^i,s_{n+1}^{-i}) - G^i(\pi_{n}^i,\pi_{n}^{-i})$.

\noindent Let  $\mathbf G$ be the convex hull in $\RR^2$ of the set $\{(G^1(s,r),G^2(s,r)) \,:\, s \in S^1 \, , \, r \in S^2 \}$ and let $\Sigma= \Delta(S^1) \times \Delta(S^2) \times \mathbf G$. We define the set-valued map $\mathbf C : \Sigma \to \Sigma$ as
\begin{equation*}
  \mathbf C(z)=\left \{(\alpha^1, \alpha^2, \gamma) \in \Sigma \, :\, \alpha^1 \in \BR^1(v^2), \, \; \alpha^2 \in \BR^2(v^1), \;  \, \gamma= (G^1(\alpha^1, \alpha^2),G^2(\alpha^1, \alpha^2))  \right \},
\end{equation*}
for $z=(v^1, v^2, \overline g^1, \overline g^2) \in \Sigma$ and we consider the differential inclusion
\begin{equation}\label{inclusion_p}
 \dot z \in -z + \mathbf C(z).
\end{equation}

\noindent  From Theorem~\ref{main2}, the map $\mathbf C$  is adapted to $(z_n, \overline \mu_n)$ $\mathbf C$, where $\overline \mu_n= (\pi_n^1, \pi_n^2, (G^1(\pi_n^1, \pi_n^2),G^2(\pi_n^1, \pi_n^2))$.  

\noindent Proposition~\ref{bruit2}~(ii) shows that $\epsilon(W_{n+1}^{i,4}/(n+1), T)$ goes to zero almost surely for all fixed $T>0$. Therefore, by writing the evolution of $z_n$ in the same manner as for $v_n$ before,  we can conclude that the limit set of the sequence $(z_n)_n$ is an ICT set of the differential inclusion \eqref{inclusion_p}.

\paragraph{Zero-sum games}

\noindent  Hofbauer and Sorin~\cite{hs06}
(by exhibiting an explicit Lyapunov function) show that the set of Nash equilibria is a global attractor for $\BR$. Hence, if we denote by $g_*$  the value of the game, a direct consequence is that $$\{(v^1, v^2,g^1, g^2) \, : \, v^1 \ \in \BR^1(v^2) \, , \, v^2 \in \BR^2(v^1) \, , \, (G^1(v^1,v^2),G^2(v^1,v^2))=(g_*, - g_*))    \}$$  is a global attractor for \eqref{inclusion_p}. Therefore $(v_n)_n$ converges to the set of Nash equilibria and $\overline g_n^1$ converges to the value of the game.

\paragraph{Potential games} In the same spirit as above, $\Phi$ is a Lyapunov function for the differential inclusion \eqref{inclusion_p} (see Bena\"im, Hofbauer and Sorin~\cite[Theorem 5.5]{bhs05}). Since, in our case, the payoff functions are linear in all variables, Propositions 3.27 and 3.28 in \cite{bhs05} imply that $(v_n)_n$ converges almost surely to a connected component of Nash equilibria on which the potential $\Phi$ is constant. In particular, if $G^1 = G^2$, let $G^*$ be the value of $G$ on the limit set of $(v_n)_n$. Then $\lim_n G(v_n^1, v_n^2) = G^*$. Therefore, by definition of $\mathbf C$, we also have $\lim_n \overline{g}_n^1 = G^*$.

\paragraph{2 $\times N$ games} Our result follows from the fact that any trajectory of the Best-Response dynamics converges to the set of Nash equilibria in this case (see Berger~\cite{berger05}).

\paragraph{\bf Acknowledgements} This paper was initially motivated by a question from  Drew Fudenberg and Satoru Takahashi. The development of this project was partially funded by Fondecyt grant No. 3130732 and  by the Complex Engineering Systems Institute (ICM: P-05-004-F, CONICYT: FBO16). The authors would like to thank the Aix-Marseille School of Economics for inviting M. Bravo to work on this project. Both authors are indebted to Sylvain Sorin for useful discussions, and for inviting M. Faure to Jussieu (Paris 6) in the early stages of this work.

\appendix
\section{Appendix}

While Assumption~\eqref{hyp_param} is used here, in fact, the proofs are written in such a way that they can be easily extended to the case where the less stringent Assumption~\eqref{hyp_param'} is considered on the sequences $(\beta_n^i)_n$.

\subsection{A general result}

Returning to the framework of Section \ref{ss:general}, we consider a discrete time stochastic process $(s_n,M_n)_n$, defined on the probability space $(\Omega, \mathcal F, \mathbb P)$, taking values in $S \times \mathcal M(S)$ and satisfying \eqref{hip0}. The  space $(\Omega, \mathcal F, \mathbb P)$ is equipped with a non-decreasing sequence of $\sigma$-algebras $(\mathcal F_n)_n$. Let $\Sigma$ be a compact convex nonempty set which is assumed, for simplicity, to be contained in $\RR^{|S|}$. It will become clear that the argument extends to the case of arbitrary euclidean spaces.

\noindent As before, let $H: S \to \Sigma$, $V_n=H(s_n)$ and $\mu_n$ be defined by \eqref{eq:mu}.  The pseudo-inverse matrix of $M_n$ is denoted by $Q_n$ (see \eqref{pseudo}). The following technical proposition will be key to our main result.
\begin{proposition}\label{prop:bruit}
Let $(\varepsilon_n)_n$ be a real random process which is adapted to $(\mathcal F_n)_n$.  Let us assume that, almost surely,
\begin{enumerate}
\item $|\varepsilon_n| |Q_n| \leq n^{a}$ for $a <1/2$ and $n$ large,
\item  $ |Q_n| |\varepsilon_{n} -\varepsilon_{n-1}| \to 0$ , 
\item $|\varepsilon_n| \left ( |Q_{n+1} - Q_n| + |\pi_{n+1} - \pi_{n}| \right)  \to 0$.
\end{enumerate}
Let $$ W_{n+1} =\varepsilon_n \left (V_{n+1} -\mu_{n} \right ).$$
Then, for all $T>0$, 
\begin{equation*}
\epsilon \left(\frac{1}{n+1} W_{n+1},T\right) \to 0,
\end{equation*}
almost surely as $n$ goes to infinity.

\end{proposition}

\begin{proof}
 Let $c$ be a positive constant that may change from line to line. In a similar manner as in the proof of \cite[Theorem 2.6]{br10}, we can decompose the noise term as follows:
\begin{align*}
\frac{1}{n+1}W_{n+1}&= \frac{\varepsilon_n}{n+1}(V_{n+1} - \mu_{n}),\\
&= \frac{\varepsilon_n}{n+1}(H(s_{n+1} )- \mu_{n}),\\
&= \frac{\varepsilon_n}{n+1}(H(s_{n+1} )- \sum_{s \in S} \pi_n(s) H(s)).
\end{align*}
 For a matrix $A \in \RR^{|S|\times |S|}$ and for any $r \in S$, let  $A[r]$ be the $r$-th line of $A$. Let us identify the function $H$ with the matrix $\mathbf H$ where, for each $r \in S$, $\mathbf H[r] = H(r)$. Notice that, by definition of the matrix $\Pi_n$, we have that $\Pi_n \mathbf H[r]= \mu_n$ for every $r \in S$. Therefore we can write
 
 \begin{equation*}
 \begin{aligned}
 \frac{1}{n+1}W_{n+1}&= \frac{\varepsilon_n}{n+1} \left ( (I - \Pi_n)\mathbf H \right )[s_{n+1}]\\
 &=  \frac{\varepsilon_n}{n+1} \left ( (Q_{n} - M_{n} Q_{n} )\mathbf H \right ) [s_{n+1}],\\
 &=  \frac{\varepsilon_n}{n+1}   \left ( (Q_{n}\mathbf H) [s_{n+1}] -   (M_{n} Q_{n} \mathbf H)  [s_{n+1}]  \right ), \\
  &= \sum_{j=1}^4u_k^j, \\
  \end{aligned}
 \end{equation*}
where the second identity follows from the definition of the pseudo-inverse matrix, and
\begin{align*}
u_n^1 &= \frac{\varepsilon_n}{n+1}\left ( (Q_{n}\mathbf H) [s_{n+1}]  - (M_{n} Q_{n}\mathbf H) [s_n] \right), \\
u_n^2 &=  \frac{\varepsilon_n}{n+1} (M_{n} Q_{n} \mathbf H) [s_n] - \frac{\varepsilon_{n-1}}{n} (M_{n} Q_{n} \mathbf H)[s_n] , \\
u_n^3 &=\frac{\varepsilon_{n-1}}{n} (M_{n} Q_{n} \mathbf H)[s_n] -  \frac{\varepsilon_n}{n+1} (M_{n+1}Q_{n+1} \mathbf H)[s_{n+1}], \\
u_n^4 &= \frac{\varepsilon_n}{n+1}(M_{n+1}Q_{n+1} \mathbf H)[s_{n+1}] -\frac{\varepsilon_n}{n+1} (M_{n}Q_{n} \mathbf H)[s_{n+1}],\\
&=\frac{\varepsilon_n}{n+1}\left (M_{n+1}Q_{n+1} - M_{n}Q_{n} \right )\mathbf H[s_{n+1}].
\end{align*}

\noindent Since $\escon{ (Q_{n}\mathbf H) [s_{n+1}]}{\mathcal F_n}= (M_{n} Q_{n}\mathbf H) [s_n]$, the random process $u_n^1$ is a martingale difference and 
$$\norm{u_n^1}\leq c \frac{|\varepsilon_n||Q_n|}{n+1}.$$
\noindent The exponential martingale inequality (see Equation~(18) in \cite{benaim99}) gives that, for all $K>0$  
\begin{equation*}
\mathbb P(\epsilon(u_n^1,T) \geq K) \leq c \expo \left ( \frac{-K^2}{c \sum_{j=n}^{m(\tau_n +T)} \varepsilon_j^2 |{Q}_j|^2/j^{2} }\right ).
\end{equation*} 
\noindent By assumption we have that, almost surely and for $j$ large enough,  $|\varepsilon_j| |Q_j| \leq  j^{a}$, for $a < 1/2$. So that 
\begin{equation*}
c \sum_{j=n}^{m(\tau_n +T)} \frac{ \varepsilon_j^2 |{Q}_j|^2}{j^{2}}  \leq c\frac{1}{n^{1 - 2a}} \sum_{j=n}^{m(\tau_n +T)} \frac{1}{j} \leq c\frac{T+1}{n^{1 - 2a}},
\end{equation*}
by definition of $m(t)$. Therefore
\begin{equation*}
  \mathbb P(\epsilon(u_n^1,T) \geq K) \leq c \expo \left ( - \frac{K^2}{T+1}n^{2a -1} \right ).
\end{equation*}
Finally, from the fact that $a< 1/2$, we have $\sum_{n\geq1} \mathbb P(\epsilon(u_k^1,T) \geq K)  < +\infty$ for all $K>0$ and the Borel--Cantelli lemma implies that $\epsilon(u_n^1,T) \to 0$ almost surely.

\noindent For the second term,  
\begin{equation*}
\begin{aligned}
\epsilon(u_n^2,T) & \leq c  \sum_{j=n}^{m(\tau_n +T)} |Q_j|\left | \frac{\varepsilon_{j-1}}{j} -  \frac{\varepsilon_j}{j+1}  \right |, \\
&=  c  \sum_{j=n}^{m(\tau_n +T)} |Q_j|\left | \frac{(j+1)\varepsilon_{j-1} - j \varepsilon_j}{j(j+1)} \right | =  c  \sum_{j=n}^{m(\tau_n +T)} |Q_j|\left |  \frac{\varepsilon_{j-1} -\varepsilon_{j}}{j} + \frac{\varepsilon_j}{j(j+1)} \right |\\
&\leq c\left [\sup_{j \geq n} |Q_j| |\varepsilon_{j} -\varepsilon_{j-1}| +  \sup_{j \geq n}\frac{ |Q_j| |\varepsilon_j|}{j+1} \right ]\sum_{j=n}^{m(\tau_n +T)} \frac{1}{j},\\
 &\leq c\left [\sup_{j \geq n} |Q_j| |\varepsilon_{j} -\varepsilon_{j-1}| +  \sup_{j \geq n}\frac{ |Q_j| |\varepsilon_j|}{j+1} \right ](T+1),\\
\end{aligned}
\end{equation*}
by definition of $m(t)$. Hence, from assumptions (i) and (ii), we conclude that $\epsilon(u_n^2,T) \to 0$ almost surely.

\noindent Now for $u_n^3$, by cancellation of successive terms, 
\begin{equation*}
\begin{aligned}
\epsilon(u_n^3,T) & = \frac{\varepsilon_{n-1}}{n} (M_{n} Q_{n} \mathbf H)[s_n] - \frac{\varepsilon_{m(\tau_n +T)-1}}{m(\tau_n +T)} (M_{m(\tau_n +T)} Q_{m(\tau_n +T)} \mathbf H) [s_{m(\tau_n +T)}], \\
& \leq 2 \sup_{j \geq n}\frac{ |Q_j| |\varepsilon_{j-1}|}{ j}.\\
\end{aligned}
\end{equation*}
which implies, by (i), that $\epsilon(u_n^3,T) \to 0$ almost surely.
\vspace{.2cm}

\noindent For the fourth term, recall that $M_nQ_n=Q_n - I + \Pi_n$ for all $n \in \NN$. Therefore, we can write
\begin{equation*}
u_n^4=\frac{\varepsilon_n}{n+1}  \left ( Q_{n+1} - Q_{n} -( \Pi_{n+1} - \Pi_{n}) \right ) \mathbf H[s_{n+1}].
\end{equation*}
Hence
\begin{equation*}
\begin{aligned}
\epsilon(u_n^4,T) &\leq  c \sum_{j=n}^{m(\tau_n +T)} \frac{1}{j} \left [ \sup_{j \geq n}|\varepsilon_j| |Q_{j+1} - Q_j| + |\varepsilon_j||\pi_{j+1} - \pi_{j}| \right ] ,\\
& \leq c (T+1)\left [ \sup_{j \geq n}|\varepsilon_j| |Q_{j+1} - Q_j| + |\varepsilon_j||\pi_{j+1} - \pi_{j}| \right ].
\end{aligned}
\end{equation*}
Assumption (iii) implies that $\epsilon(u_n^4,T) \to 0$ almost surely, as $n$ goes to infinity.
\end{proof}

\subsection{Stability}

\noindent The following lemma is a trivial consequence of the recursive definition of the vector $R_n^i$ and the fact that $\gamma_n^i(s) \in ]0,1]$.

\begin{lemma}  \label{lm:R} For any $i \in \{1,2\}$, $s \in S^i$, $n \in \mathbb{N}$, we have $R_n^i(s) \in [-K^i,K^i]$, where 
  \begin{equation}
    \label{eq:ki}
K^i = \max\{ \max_{s,r \in S^i}|R_0^i(s) - R_0^i(r)|, \max_{s \in S^i}\max_{s^{-i}, r^{-i} \in S^{-i}} |G^i(s,s^{-i}) - G^i(s, r^{-i})|\}.
  \end{equation}
\end{lemma}
\noindent The following result states that, without loss of generality, we can suppose that the step size $\gamma_n^i(s)$ is equal to $1/n \pi_{n-1}^i(s)$ for all $s \in S^i$.

\begin{lemma}\label{stabilite} Let $\alpha \in ]0,1[$
There exists $n_0(\alpha) \in \mathbb{N}$ \textup{(}which only depends on $\alpha$, $R^i_0$, the payoff functions $G^i$ and the vanishing sequence $(A_n^i)_n$\textup{)} such that, for any $n \geq n_0(\alpha)$ and $s \in S^i$, 
\[\pi_n^i(s) \geq    n^{- \alpha}.\]

In particular, there exists $n_0 \in \mathbb{N}$ such that, for any $n \geq n_0$ and $s \in S^i$, 
\[\gamma_n^i(s) = \frac{1}{n \pi_{n-1}^i(s)}.\]
\end{lemma}
\begin{proof} Let $\alpha \in ]0,1[$ and $\alpha' \in ]0,\alpha[$. Choose $n_0 \in \mathbb{N}$ such that, for any $n \geq n_0$, $2K^i A_n^i  \leq \alpha'$, where $K^i$ is defined in \eqref{eq:ki}, and take $r_n \in S^i$ such that $R_n^i(r_n)= \max_r R_n^i(r)$. Then, for any $s \in S^i$,
\begin{equation*}
\begin{aligned}
\pi_n^i(s)&= \dfrac{\pi_0^i(s)\exp \left ( A_n^i \ln(n) (  R_n^i(s) - R_n^i(r_n) ) \right )}{\pi_0^i(r_n) + \sum_{r \neq r_n} \pi_0^i(r)\exp \left ( A_n^i \ln(n) (  R_n^i(r) - R_n^i(r_n)  \right ) }\\
&\geq \pi_0^i(s)  \exp \left (-2A^i_n \ln(n)K^i \right ) \geq  \min_r \pi_0^i(r) \exp (-\alpha' \ln(n)) \geq  \min_r \pi_0^i(r)  n^{- \alpha'}. \\
\end{aligned}
\end{equation*}
Without loss of generality, we can assume that $n_0$ is large enough so that $\min_r \pi_0^i(r)  n^{- \alpha'} \geq n^{- \alpha}$. This concludes the proof of the first point. In particular, there exists $n_0 \in \mathbb{N}$ such that, for any $n \geq n_0$, $(n+1) \pi_n^i(s) > 1$, which proves the second point. 
\end{proof}

\subsection{Analysis of the noise sequences}\label{est_noise}
Let us fix $i \in \{1,2\}$ and let $\chi_n^i$ be the spectral gap of the matrix $M_n^i=M[\beta_n^i, R_n^i]$, i.e.
\[\chi_n^i = \min \biggl \{ \frac{\mathcal E_n^i(f,f)}{\var^i_n(f)} \, :\, \var_n^i(f)\neq 0 \biggr \},\]
where 
\begin{align*}
\var_n^i(f)&= \sum_{s \in S^i} \pi_n^i(s)f^2(s) - \left ( \sum_{s \in S^i} \pi_n^i(s)f(s)\right )^2,\\
\mathcal E_n^i(f,f) &= \frac{1}{2} \sum_{s,r \in S^i} (f(s) - f(r))^2M_n^i(s,r)\pi_n^i(s).
\end{align*}

\noindent The following result is a direct consequence of results of Holley and Strook~\cite{HolleyStroock88}. 
\begin{lemma} \label{lemma:spectral}
 There exists a positive constant $c$  such that, for a sufficiently large $n \in \NN$
 \begin{equation}\label{spectral}
 c \expo(- 2K^i \beta_n^i )\leq \chi_n^i,
 \end{equation}
where $K^i$ is defined in Equation \eqref{eq:ki}.
\end{lemma}

\begin{proof} By Lemma 2.7 in \cite{HolleyStroock88}, for sufficiently large $n$,
\[\chi_n^i \geq c \expo(- \beta_n^i \mathbf{m}_n),\]
where 
\[\mathbf{m}_n = \max_{s,r \in S^i} \left\{\min_{\gamma \in \Gamma} \max_{s' \in \gamma} R_n^i(s') - R_n^i(s) - R_n^i(r) + \min_{s' \in S^i} R_n^i(s') \right\},\]
and $\Gamma$ is the set of every path from $s$ to $r$ on the graph that represents the action set of player $i$. Now it is clear that $\mathbf{m}_n \leq 2K^i$, by Lemma \ref{lm:R}. 
\end{proof}

\begin{lemma} \label{vitesse_matrices:2} Given $s \in S^i$, under assumption \eqref{hyp_param} the following holds, almost surely, as $n \to +\infty$. 
\begin{itemize}
 \item[\textup{(i)}] $\dfrac{|Q_{n}^i|}{n^{a} \pi_n^i(s)^b} \longrightarrow 0$, for any $a >0, b >0$,
\item[\textup{(ii)}] $\dfrac{|Q_{n+1}^i - Q_n^i|n^{1-\alpha}}{ \pi_n^i(s)}\longrightarrow 0$ and $\dfrac{|\pi_{n+1}^i - \pi_n^i|n^{1-\alpha}}{\pi_n^i(s)}\longrightarrow 0$ for any $\alpha >0$.
\end{itemize}
\end{lemma}
\begin{proof} Let $c$ be a general positive constant that may change from line to line.
\begin{enumerate}
\item  The first inequality in \cite[Proposition~3.4]{br10} (based on estimations obtained by Saloff-Coste~\cite{sc97}) reads in this case, for $n \in \NN$ and $s,s' \in S^i$,
  \begin{equation} \label{eq:SG}
    |Q_n^i(s,s')| \leq \frac{1}{\chi_n^i} \left(\frac{\pi_n^i(s')}{\pi_n^i(s)}\right)^{1/2} \leq \frac{1}{\chi_n^i} (\pi_n^i(s))^{-1/2},
  \end{equation}

\noindent Let $a>0$ and $b >0$. By Lemma~\ref{lemma:spectral}, $(\chi_{n}^i)^{-1} \leq c^{-1} n^{2K^i A_n^i }$. Pick $\frac{a}{b + 1/2} >\alpha >0$. There exists $n_0(\alpha)$ such that, for any $n \geq n_0$, for any $s \in S^i$, $\pi_n^i(s) \geq  n^{-\alpha}$  Therefore for sufficiently large $n$,
\begin{equation*}
\frac{|Q_{n}^i|}{n^{a}\pi_n^i(s)^b} \leq c^{-1} \frac{n^{2K^i A_n^i + \alpha/2}}{n^a n^{- b\alpha}} = c^{-1} n^{2 K^i A_n^i + \alpha (1/2 + b) - a}.
\end{equation*}
Thus the conclusion follows from the fact that $\alpha(1/2 + b) - a < 0$ and $\lim_n A_n^i = 0$.
\item Let $\alpha >0$. Recall that $M_n^i=M^i[\beta_n^i,R_n^i]$ and let us assume without loss of generality that the sequence $(A_n^i)_n$ is non-increasing. Therefore
\begin{equation*}
 |M_{n+1}^i - M_{n}^i| \leq \left |M^i[\beta_{n+1}^i,R_{n+1}^i] - M^i[\beta_n^i,R_{n+1}^i] \right | + \left | M^i[\beta_{n}^i,R_{n+1}^i] - M^i[\beta_{n}^i,R_{n}^i] \right |.
\end{equation*}
\noindent A simple application of the mean value theorem on the functions $\beta \to M^i[\beta, R]$ and $R \to M^i[\beta, R]$ yields, respectively,
\begin{equation*}
\left  |M^i[\beta_{n+1}^i,R_{n+1}^i]- M^i[\beta_{n}^i,R_{n+1}^i]\right | \leq \frac{A_n^i}{n},
\end{equation*}
and
\begin{equation*}
\left |M^i[\beta_{n}^i,R_{n+1}^i]-M^i[\beta_{n}^i, R_{n}^i] \right | \leq c \beta_n^i |R_{n+1}^i - R_n^i|,
\end{equation*}

By Lemma~\ref{stabilite}, and since $|R_{n+1}^i - R_n^i| \leq \max_{s \in S^i} c \gamma_{n+1}^i(s)$, we have that 
\[\left | M^i_{n+1} - M^i_{n}\right | \leq  \frac{1}{n^{1-\alpha/4}}\]
for sufficiently large $n$. Analogously, recalling that $\pi_n^i=\pi^i[\beta_n^i,R_n^i]$, we have
\begin{equation} \label{pi:vitesse}
|\pi_{n+1}^i -\pi_{n}^i| \leq  \frac{1}{n^{1- \alpha/4}}.
\end{equation}
for sufficiently large $n$.
\noindent Recall that, from part (i), $|Q_n^i| \leq n^{\alpha/8}$ for sufficiently large $n$. Also, $\pi_n^i(s) \geq n^{-\alpha/4}$.  Using the last inequality in the proof of \cite[Proposition 3.3]{br10}:  
\[|Q_{n+1}^i - Q_n^i| \leq c \left(|Q_{n+1}^i||Q_n^i||M_{n+1}^i-M_{n}^i| + |Q_n^i||\pi_{n+1}^i-\pi_n^i|\right),\]
we have that 
\begin{align*}
\frac{|Q_{n+1}^i - Q_n^i|n^{1- \alpha}}{\pi_n^i(s)} &\leq c n^{1- \alpha} n^{\alpha/4} \left(|Q_{n+1}^i||Q_n^i||M_{n+1}^i-M_{n}^i| + |Q_n^i||\pi_{n+1}^i-\pi_n^i|\right)\\
& \leq \frac{1}{n^{\alpha/8}},
\end{align*}
almost surely, for sufficiently large $n$. 
\end{enumerate}
\end{proof}

\noindent The following two propositions establish all the results on the noise terms that we need in the proof of Theorem~\ref{main2} ({\it c.f.} Section~\ref{proofs}).

\begin{proposition}\label{bruit}
Assume that \eqref{hyp_param} holds and let $i \in \{1,2\}$.
\begin{enumerate}
\item For $s \in S^i$, let 
\begin{equation} \label{ruido1}
 W^{i,1}_{n+1}(s) =\frac{R_n^i(s)}{\pi_{n}^i(s)} \left ( \ind_{\{s_{n+1}^i=s \}} -\pi_{n}^i(s) \right ) \in \RR. 
\end{equation}
Then, for all $T>0$, 
\begin{equation*}
\epsilon \left(\frac{1}{n+1} W^{i,1}_{n+1}(s),T\right) \to 0,
\end{equation*}
almost surely as $n$ goes to infinity.

\item Let 
$$ \overline W^{i}_{n+1}=\delta_{s_{n+1}^i}-\pi_{n}^i \in \RR^{|S^i|} .$$
Then, for all $T>0$, 
\begin{equation*}
\epsilon \left(\frac{1}{n+1} \overline W^{i,1}_{n+1},T\right) \to 0,
\end{equation*}
almost surely as $n$ goes to infinity.
\item Let 
$$W^{i,3}_{n+1} = G^i(\cdot ,s_{n+1}^{-i}) -G^i(\cdot,\pi_{n}^{-i}) \in \RR^{|S^i|}.$$
Then, for all $T>0$, 
\begin{equation*}
\epsilon \left(\frac{1}{n+1} W^{i,3}_{n+1},T \right) \to 0,
\end{equation*}
almost surely as $n$ goes to infinity.
\end{enumerate}
\end{proposition}
\begin{proof}
We prove part (i) in detail. Given that the arguments are very similar, the remaining proofs are omitted.

\noindent  We  apply Proposition~\ref{prop:bruit} with $S= S^i$, $\Sigma= \Delta(S^i)$, $s_n=s_n^i$, $M_n= M_n^i$, $\pi_n=\pi_n^i$ and $H(r)= \delta_r$ for all $r \in S^i$. Therefore in this case $\mu_n = \pi_n^i$ and $V_{n+1}= \delta_{s_{n+1}^i}$. We also put $\varepsilon_n= R_n^i(s)/\pi_n^i(s)$.  

\noindent From the fact that $R_n^i$ is bounded, it is easy to see that points $\text{(i)}$ and $\text{(ii)}$ of Lemma \ref{vitesse_matrices:2} respectively imply assumptions $\text{(i)}$ and $\text{(iii)}$ of Proposition~\ref{prop:bruit}.  To confirm that assumption (ii) holds, it suffices to compute
\begin{equation}
\begin{aligned} \label{eq:}
|Q_n^i||\varepsilon_n - \varepsilon_{n-1}|&= \dfrac{|Q_n^i||\pi_{n}^i(s)( R^i_{n}(s)- R^i_{n-1}(s)) + R^i_{n}(s)(\pi_{n-1}^i(s) - \pi_{n}^i(s) )|}{\pi_{n}^i(s)\pi_{n-1}^i(s)}\\
&\leq c |Q_n^i| n^{-1 +\alpha}
\end{aligned}
\end{equation}
by definition of $R_n^i$, Lemma~\ref{stabilite} and equation~\eqref{pi:vitesse}, for sufficiently large $n$ and any $\alpha>0$. Hence, by Lemma~\ref{vitesse_matrices:2}, $|Q_n^i||\varepsilon_n - \varepsilon_{n-1}|$ goes to zero almost surely as $n$ goes to infinity.
\noindent By using Proposition~\ref{prop:bruit}, we show that $\epsilon( \mathcal U^i_{n+1}/(n+1), T)$ goes to zero almost surely for any $T>0$, where

$$\mathcal U^i_{n+1}=\frac{R_n^i(s)}{\pi_{n}^i(s)} \left ( \delta_{s_{n+1}^i } - \pi_n^i \right ) \in \RR^{|S^i|}.$$
The result follows from the fact that the $s$-th component of the vector $\mathcal U_{n+1}^i$ is equal to  $W^{i,1}_{n+1}(s)$.

\end{proof}

\begin{proposition}\label{bruit2}
Assume that \eqref{hyp_param} holds and let us fix $i \in \{1,2\}$.

\begin{enumerate}
\item For $s \in S^i$, let 
\begin{equation*}
 W^{i,2}_{n+1}(s) =\frac{1}{\pi^i_{n}(s)} \left ( \ind_{\{s_{n+1}^i=s \}} G^i( s_{n+1}^1, s_{n+1}^2) -\pi_{n}^i(s) G^i(s, \pi^{-i}_{n}) \right ) \in \RR.
\end{equation*}
Then, for all $T>0$, 
\begin{equation*}
\epsilon \left(\frac{1}{n+1} W^{i,2}_{n+1},T\right) \to 0,
\end{equation*}
almost surely as $n$ goes to infinity.
\item Let 
\begin{equation*}
 W^{i,4}_{n+1} = G^i( s_{n+1}^i, s_{n+1}^{-i}) - G^i(\pi_n^i, \pi^{-i}_{n}) \in \RR.
\end{equation*}
Then, for all $T>0$, 
\begin{equation*}
\epsilon \left(\frac{1}{n+1} W^{i,4}_{n+1},T\right) \to 0,
\end{equation*}
almost surely as $n$ goes to infinity.
\end{enumerate}

\end{proposition}
\begin{proof}
  \begin{enumerate}
  \item For the sake of clarity, let us set $i=1$. Again, we use Proposition~\ref{prop:bruit}, where in this case,  $S= S^1 \times S^2$, $\Sigma$ is defined by 
\[\Sigma= \left\{ \sum_{s^2 \in S^2}\sigma^2(s^2)G^1(\cdot, s^2) \, : \, \sum_{s^2 \in S^2} \sigma^2(s^2) = 1 \text{ and } \sigma^2(s^2) \geq 0 \text{ for all } s^2 \in S^2  \right\} \subseteq \RR^{|S^1|},\]
$s_n=(s_n^1, s_n^2)$, $M_n= M_n^1 \otimes M_n^2$, $\pi_n= \pi_n^1 \otimes \pi_n^2$ and $H : S^1\times S^2  \to \Sigma$ where 
$$H(s^1,s^2)= \delta_{s^1} G^1(s^1,s^2),$$ 
for all $(s^1,s^2) \in S^1 \times S^2$. Notice that in this case $\delta$ is the Kronecker's delta function taking values in $\Delta(S^1)$. Therefore in this case $\mu_n=(\mu_n(s^1))_{s^1 \in S^1}$, with $\mu_n(s^1) = \pi_n^1(s^1) G^1(s^1, \pi_n^2)$ and $V_{n+1}= (V_{n+1}(s^1))_{s^1 \in S^1}$, where
$$V_{n+1}(s^1)= \ind_{\{s_{n+1}^1=s^1\}}G^1(s_{n+1}^1, s_{n+1}^2)=  \ind_{\{s_{n+1}^1=s^1 \}}G^1(s^1, s_{n+1}^2).$$ 

We also set in this case $\varepsilon_n= 1/\pi_n^1(s)$.  Let $Q_n$ be the pseudo-inverse matrix of the stochastic matrix $M_n$.  It is easy to see that the spectral gap of $M_n$ verifies that
\begin{equation*}
\chi(M_n)= \chi(M_n^1 \otimes M_n^2)= \min \{\chi(M_n^1), \chi(M_n^2) \}=\min \{\chi_n^1, \chi_n^2 \}.
\end{equation*}

\noindent By using inequality \eqref{eq:SG} for the matrix $Q_n$ and the fact that $\pi_n(s^1,s^2)=\pi_n^1(s^1)\pi_n^2(s^2) \geq n^{-\alpha}$ for any $\alpha>0$ and sufficiently large $n$, we can obtain exactly the same conclusions as in Lemma~\ref{vitesse_matrices:2} for $Q_n$ and $\pi_n$. 

\noindent Hence, as in the proof of Proposition~\ref{bruit}, we deduce that sequences $(\varepsilon_n)_n$ and $(Q_n)_n$ verify assumptions $\text{(i)-(iii)}$ of Proposition~\ref{prop:bruit}.  

\noindent Therefore,  we have that $\epsilon( \mathcal U^i_{n+1}/(n+1), T)$ goes to zero almost surely for any $T>0$ where, for $s^1 \in S^1$,
\begin{equation*}
\begin{aligned}
\mathcal U^i_{n+1}(s^1) &=\frac{1}{\pi_n^1(s)} \left  (\ind_{\{s_{n+1}^1=s^1 \}}G^1(s^1, s_{n+1}^2) - \pi_{n}^1(s^1) G^1(s^1, \pi^{2}_{n}) \right )\\
&=\frac{1}{\pi_n^1(s)} \left  (\ind_{\{s_{n+1}^1=s^1 \}}G^1(s_{n+1}^1, s_{n+1}^2) - \pi_{n}^1(s^1) G^1(s^1, \pi^{2}_{n}) \right ).
\end{aligned}
\end{equation*} 
\noindent The conclusion follows taking $s^1=s$ in the equation above.
\item The proof of this part also follows from Proposition~\ref{prop:bruit}, taking as $\Sigma$ a sufficiently large compact set in $\RR$, $s_n=(s_n^1, s_n^2)$, $M_n= M_n^1 \otimes M_n^2$, $\pi_n= \pi_n^1 \otimes \pi_n^2$ and $H : S^1\times S^2  \to \Sigma$, where 
$$H(s^1,s^2)= G^i(s^1,s^2).$$ 

\noindent Therefore $\mu_n= G^i(\pi_n^i, \pi^{-i}_{n})$ and $\varepsilon_n=1$ for all $n \in \NN$. Finally, using the same argument as in part (i), we prove that the assumptions (i)-(iii) hold and we conclude.
\end{enumerate}
\end{proof}
\bibliography{biblio/pomfp}
\bibliographystyle{biblio/ims}
\end{document}